\documentclass{article} %podeu treure el twocolumn si voleu una columna
\usepackage[utf8]{inputenc}
\usepackage{graphicx}
\usepackage{fancyhdr}
\usepackage{lastpage}
\usepackage{geometry}
\usepackage{lipsum}
\usepackage{relsize}
\usepackage{listings}
\usepackage{color}
\usepackage{amsmath}
\usepackage{bbm}
\usepackage{amsthm}
\usepackage{enumerate}
\usepackage{array}
\usepackage{float}
\usepackage{hyperref}
\usepackage{cite}
\usepackage{multirow}
\usepackage{amsmath}
\usepackage{amssymb}
\usepackage[colorinlistoftodos]{todonotes}
\hypersetup{
	colorlinks=true,       % false: boxed links; true: colored links
	linkcolor=black,        % color of internal links
	citecolor=black,        % color of links to bibliography
	filecolor=magenta,     % color of file links
	urlcolor=blue         
}

\newgeometry{left=2.7cm, right=2.7cm, top=3.5cm, bottom=3.2cm}
\newtheorem{theorem}{Theorem}
\newtheorem{lemma}{Lemma}
\newtheorem{example}{Example}

\newtheorem{definition}{Definition}

\usepackage{listings}
\usepackage{color}

\definecolor{dkgreen}{rgb}{0,0.6,0} %canvieu els colors
\definecolor{gray}{rgb}{0.5,0.5,0.5} % pels que us agradin!
\definecolor{mauve}{rgb}{0.58,0,0.82}

\title{Smoothing Linear Codes by Rényi Divergence and Applications to Security Reduction }
\date{} %canviar si voleu data!
\begin{document}

\begin{titlepage}
\author{ Hao Yan, Cong Ling \\ Imperial College London\\  \texttt{h.yan22@imperial.ac.uk}, \texttt{c.ling@imperial.ac.uk} }
\maketitle

\thispagestyle{empty}
\begin{abstract}
 The concept of the smoothing parameter plays a crucial role in both lattice-based and code-based cryptography, primarily due to its effectiveness in achieving nearly uniform distributions through the addition of noise. Recent research by Pathegama and Barg has determined the optimal smoothing bound for random codes under Rényi Divergence for any order $\alpha \in (1, \infty)$ \cite{pathegama2024r}. Considering the inherent complexity of encoding/decoding algorithms in random codes, our research introduces enhanced structural elements into these coding schemes. Specifically, this paper presents a novel derivation of the smoothing bound for random linear codes, maintaining the same order of Rényi Divergence and achieving optimality for any $\alpha\in (1,\infty)$. We extend this framework under KL Divergence by transitioning from random linear codes to random self-dual codes, and subsequently to random quasi-cyclic codes, incorporating progressively more structures. As an application, we derive an average-case to average-case reduction from the Learning Parity with Noise (LPN) problem to the average-case decoding problem. This reduction aligns with the parameter regime in \cite{debris2022worst}, but uniquely employs Rényi divergence and directly considers Bernoulli noise, instead of combining ball noise and Bernoulli noise.
\end{abstract}
\end{titlepage}

\newpage
\tableofcontents

%%%%%%%%%%%%%%%%%%%%%%%%%%%%%%%%%%%%%%%%%%%%%%%%%%%%%%%
%INTRODUCCIO
%%%%%%%%%%%%%%%%%%%%%%%%%%%%%%%%%%%%%%%%%%%%%%%%%%%%%%%

\section{Introduction}

\par{\textbf{Smoothing Parameter.}} 
In the context of lattices or codes, smoothing refers to the phenomenon where adding sufficiently large noise causes the final distribution to approximate uniformity over the entire Euclidean or Hamming space. The smoothing parameter is defined as the threshold noise level associated with a specific approximation error, \(\epsilon\). Conceptually, it quantifies the degree of "smoothness" required in the noise distribution to ensure the discrete structure of the lattice \(\Lambda\) or code \(C\) becomes indiscernible. This parameter is pivotal in transforming lattice or code decoding problems into cryptographic security proofs. Specifically, in lattice-based cryptography, the security of problems such as the Short Integer Solution (SIS) and Learning With Errors (LWE) is reduced to the difficulty of lattice problems involving the smoothing parameter, as measured by Statistical Distance (SD) \cite{micciancio2007worst}\cite{regev2009lattices}\cite{peikert2009public}\cite{brakerski2013classical}, establishing a foundation for post-quantum cryptography. Smoothing parameters can be also used in other lattice based problems such as lattice isomorphism problem \cite{ducas2022lattice}. Additionally, Rényi divergence serves as another metric for evaluating approximation error within cryptographic contexts \cite{bai2018improved}\cite{prest2017sharper}.

The smoothing parameter, a critical tool in cryptography, has also been optimized by researchers in the information theory community under the concept of channel resolvability.  Channel resolvability addresses the problem of determining the amount of information required to simulate a given channel and its output. Initially proposed by Han and Verdú \cite{han1993approximation}, this concept used SD and Kullback-Leibler (KL) divergence to measure approximation error. Hayashi later extended the solution to the KL divergence framework \cite{hayashi2006general}\cite{hayashi2011exponential}, while Yu and Tan generalized it to the Rényi divergence for parameters \(\alpha \in [0, 2] \cup \{\infty\}\) \cite{yu2018renyi}. For \(\alpha \in (2,\infty)\), Pathegama and Barg addressed the problem using random codes, and even proposed that Reed-Muller code can achieve the optimal bound under Bernouli noise, albeit limited to uniform distribution targets \cite{pathegama2023smoothing}. Yu provided a more comprehensive solution \cite{yu2024renyiresolvability}, presenting an achievable coding scheme through random codes, constant composition codes and typical sets.

\par{\textbf{Code-Based Cryptography.}} 
Coding theory is an important subject aiming to correct errors, where usually linear codes are originally used in digital communication and data transmission that contains encoding and decoding processes. It is computationally difficult to decode messages, and to decrease the decoding algorithm complexity and improve the coding efficiency, various code structures and decoding algorithms were proposed. 
 
The first code-based cryptosystem was proposed by McEliece \cite{mceliece1978public}, where binary Goppa code is used for encryption. The security of McEliece's cryptosystem relies on hardness of decoding codewords and distinguishing random matrices and permuted generator matrices. Following this work, various codes are used to build up cryptosystem, among which the MDPC-code based scheme BIKE \cite{aragon2022bike} and binary Goppa code based scheme Classic McEliece \cite{bernstein2017classic} become fourth round candidates in the NIST call for PQC standarization, which are the only two schemes of McEliece's framework.

 There are also numerous work aiming to improve efficiencies of McEliece's scheme,  leading to the development of alternative frameworks. In 2003, Aleknovich proposed a new framework with a security proof that it's solely based on decoding problem \cite{alekhnovich2003more}. Aleknovich's cryptosystem can be regarded based on the code version of LWE, i.e. Learning Parity With Noise (LPN).

 Similar to R-LWE which improve efficiency by ring structures, the Ring-LPN problem was also proposed to improve efficiency of LPN, and was used for a authentication protocol \cite{heyse2012lapin}. Several years later, the HQC scheme was proposed \cite{melchor2018HQC}, representing a specific version of Module-LPN. The security of the HQC scheme is grounded in the Quasi-Cyclic Syndrome Decoding problem. Notably, the HQC scheme stands out as the sole candidate within the non-McEliece framework during the fourth round of the Post-Quantum Cryptography (PQC) standardization process. Moreover Quang Dao and Aayush Jain proposed a new variant of LPN named Dense-Sparse LPN based on density of matrices \cite{dao2024lossyDensesparseLPN}. Roughly, the assumption states that $(\mathbf{TM}, \mathbf{s} \mathbf{TM} + \mathbf{e})$ is indistinguishable from $(\mathbf{TM}, \mathbf{u})$, for a random (dense) matrix $\mathbf{T}$, random sparse matrix $\mathbf{M}$, and sparse noise vector $\mathbf{e}$ drawn from the Bernoulli distribution with inverse polynomial noise probability. In addition to the existing frameworks centered around syndromes and decoding, a novel approach has emerged, grounded in the complexity of determining the linear isometry, an equivalence transformation that preserves the metric, between two codes, which shares similarities with LIP in lattice version. Code Equivalence problem was first studied in coding theory but it wasn't until 2020 that the first cryptographic scheme exclusively relying on the hardness of the Linear Code Equivalence problem (LEP) was introduced in \cite{biasse2020less}, named Linear Equivalence Signature Scheme (LESS). 

Despite the reliance of code-based cryptographic schemes on hard problems, many of them have faced lack of complete security reduction proof from worst case to average case and search problem to decision problem. In contrast, most problems in lattice-based cryptography have been addressed. Given the shared characteristics of codes and lattices, smoothing techniques have been applied to address the LPN problem. The first worst-case to average-case reduction for codes was established in \cite{brakerski2019worst} by smoothing a code with random walk noise, which is similar to LWE's classical reduction in \cite{regev2009lattices}, albeit with a requirement for balanced codes. This reduction has been further optimized in subsequent works \cite{yu2021smoothing} and \cite{debris2022worst} where the authors proposed an average-case to average-case reduction for LPN problem as well.  However, the need for balanced codes persists because terms involving codewords of both low and high Hamming weights cannot be eliminated as $n$ approaches infinity when estimating the smoothing bounds. As for code equivalence problem, up until now there is no worst-to-average reduction for such problems.

\par{\textbf{Main Contributions.}} 
In this paper we derive the smoothing bound for random linear codes by Rényi divergence all Rényi parameters $\alpha\in (1, \infty)$. To introduce more structure into the random code, we reduce random linear codes to a class of random self-dual codes and further reduce them to a class of random quasi-cyclic codes. As an application of the smoothing bound in code-based cryptography, we derive a average-case to average-case reduction from LPN to average-case decoding problem, which has the same parameter regime in \cite{debris2022worst} but our reduction utilizes Rényi divergence and consider directly Bernoulli noise instead of combining ball noise and Bernoulli noise together.

 In Section \ref{sec:Preliminary}, some notations and preliminaries of coding theory and cryptography will be introduced. Random linear code smoothing bound for Rényi parameters $\alpha\in(1,\infty)$ will be given in Section \ref{sec:Random_Linear_Code_Smoothing}. Section \ref{sec:Random_Self_Dual_Code_Smoothing} and \ref{sec:Random_Quasi_Cyclic_Code_Smoothing} will give the smoothing bound of random self dual code and random quasi cyclic code. Application of smoothing bound will be given in Section \ref{sec:Average_to_Average_Case_Reduction}.

 As we were finalizing this paper, we became aware of the independent and concurrent work by Pathegama and Barg \cite{pathegama2024r}. It is important to distinguish our paper from \cite{pathegama2024r} in the following ways: (a) While \cite{pathegama2024r} addresses the problem of hash functions, our focus is on code-based cryptography; (b) Our results are more general, as we are capable of handling all real-valued parameters $\alpha \in (1, \infty)$, whereas in \cite{pathegama2024r}, $\alpha$ is restricted to natural numbers $\alpha \in \mathbb{N}$; (c) We study self-dual codes, particularly quasi-cyclic codes, which are widely used in the practice of code-based cryptography. This aspect also distinguishes our work from \cite{pathegama2024r} which is restricted to random linear codes.

\section{Preliminary}
\label{sec:Preliminary}

We first given some basic notations that used in this paper. For a finite set \(\mathcal{X}\) of outcomes, a probability distribution \(P\) assigns a probability \(P(x)\) to each outcome \(x \in \mathcal{X}\) such that \(0 \leq P(x) \leq 1\) for all \(x \in \mathcal{X}\) and \(\sum_{x \in \mathcal{X}} P(x) = 1\). The probability mass function (PMF) \(P(x)\) specifies the probability that a discrete random variable \(X\) takes the value \(x\). A finite field, also known as a Galois field, is a field with a finite number of elements. The number of elements in a finite field is called its order, and it is always a power of a prime number. The finite field with \(q\) elements is denoted by \(\mathbb{F}_q\). For instance, \(\mathbb{F}_2\) is the finite field with two elements, typically \(\{0, 1\}\), with addition and multiplication defined modulo 2, which can be written as $\mathbf{Z}_{2} $. Denote channel as $W(\cdot |\cdot)$. $W(y|x)$ given input $x$ and output $y$.

\subsection{Rényi Entropy and Divergence}

Rényi entropy and Rényi divergence are fundamental concepts in information theory that generalize the classical Shannon entropy and Kullback-Leibler (KL) divergence. These measures incorporate a parameter \(\alpha\) that allows for a family of entropy and divergence measures, each with different sensitivity to the probability distributions' differences. Rényi entropy of order \(\alpha\), where \(\alpha > 0\) and \(\alpha \neq 1\), for a discrete probability distribution \(P\) over a finite set \(\mathcal{X}\), is defined as:

\[ H_\alpha(P) = \frac{1}{1 - \alpha} \log \left( \sum_{x \in \mathcal{X}} P(x)^\alpha \right) \]

Here, \(\mathcal{X}\) denotes the set of possible outcomes, and \(P(x)\) represents the probability of the outcome \(x\) under the distribution \(P\). As \(\alpha \to 1\), Rényi entropy converges to the Shannon entropy:

\[ H_\alpha(P) \to H(P) = - \sum_{x \in \mathcal{X}} P(x) \log P(x) \]

Rényi entropy provides a spectrum of entropy measures.

Rényi divergence measures the difference between two probability distributions \(P\) and \(Q\). Given two discrete probability distributions \(P\) and \(Q\) over a finite set \(\mathcal{X}\), the Rényi divergence of order \(\alpha\), where \(\alpha > 0\) and \(\alpha \neq 1\), is defined as:

\[ D_\alpha(P \| Q) = \frac{1}{\alpha - 1} \log \left( \sum_{x \in \mathcal{X}} P(x)^\alpha Q(x)^{1-\alpha} \right) \]

As \(\alpha \to 1\), Rényi divergence converges to the KL divergence:

\[ D_\alpha(P \| Q) \to D_{KL}(P \| Q) = \sum_{x \in \mathcal{X}} P(x) \log \frac{P(x)}{Q(x)} \]

For \(\alpha > 0\) and \(\alpha \neq 1\), \(D_\alpha(P \| Q) \geq 0\), with equality if and only if \(P = Q\). Additionally, for \(\alpha \leq \beta\), \(D_\alpha(P \| Q) \leq D_\beta(P \| Q)\). This property makes Rényi divergence a useful tool for controlling the trade-off between robustness and sensitivity in various applications. The choice of \(\alpha\) influences how the divergence measures the difference between \(P\) and \(Q\); for \(\alpha > 1\), the divergence is more sensitive to the regions where \(P(x)\) is larger than \(Q(x)\), while for \(0 < \alpha < 1\), it is more sensitive to the regions where \(P(x)\) is smaller than \(Q(x)\).

 By adjusting the parameter \(\alpha\), researchers and practitioners can tailor the Rényi divergence to meet the specific requirements of their applications, making it a versatile and powerful tool in information theory and beyond.  If we want to consider q-ary code, the base of log can be replaced with $q$.

\subsection{Linear, Self-Dual, and Quasi-Cyclic Codes}
Since we are interested in smoothing of linear code, self dual code and quasi cyclic code, we will first introduce the concepts of these codes.

\begin{definition}[Linear Code]
A \textit{linear code} \( C \) is a subspace of the vector space \( \mathbb{F}_q^n \), where \( \mathbb{F}_q \) is a finite field with \( q \) elements. For binary codes, \( q = 2 \), thus \( \mathbb{F}_2 \). The code \( C \) is characterized by its dimension \( k \) and length \( n \), and is typically denoted as an \([n, k]\) code.
\end{definition}

Linear codes possess several key properties. Firstly, as a subspace, they ensure that for any codewords \( \mathbf{c}_1, \mathbf{c}_2 \in C \) and any scalars \( a, b \in \mathbb{F}_q \), the linear combination \( a\mathbf{c}_1 + b\mathbf{c}_2 \) is also in \( C \). This closure under addition and scalar multiplication is fundamental to their structure. The dimension \( k \) of the code represents the number of linearly independent codewords, signifying the amount of information that can be encoded, with \( C \) having \( q^k \) distinct codewords.

A generator matrix \( G \) for an \([n, k]\) linear code is a \( k \times n \) matrix whose rows form a basis for \( C \). Any codeword \( \mathbf{c} \in C \) can be expressed as \( \mathbf{c} = \mathbf{u}G \), where \( \mathbf{u} \in \mathbb{F}_q^k \) is an information vector. Encoding a message vector \( \mathbf{u} \) into a codeword \( \mathbf{c} \) involves multiplying \( \mathbf{u} \) by the generator matrix \( G \). Conversely, a parity-check matrix \( H \) for an \([n, k]\) linear code is an \((n-k) \times n\) matrix that defines the dual code \( C^\perp \). A vector \( \mathbf{v} \in \mathbb{F}_q^n \) is a valid codeword of \( C \) if and only if it satisfies the parity-check equation \( H\mathbf{v}^T = \mathbf{0} \). The rate $R$ of a linear code is defined as $R=\frac{k}{n}$.

\begin{definition}[Self-Dual Code]
    A linear code $C$ is \emph{self-dual} if $C = C^\perp$, where $C^\perp$ is the dual code of $C$, defined as the set of all vectors orthogonal to every codeword in $C$.    
\end{definition}

In this paper we only consider binary self-dual code, and in this case the dimension \( k \) is \( \frac{n}{2} \) because \( C \) and \( C^\perp \) must span the entire space \( \mathbb{F}_2^n \). The generator matrix \( G \) and the parity-check matrix \( H \) of a self-dual code are related by \( G = H \). This implies that all codewords in a self-dual code are orthogonal to each other with respect to the standard dot product, i.e. $\left \langle  \mathbf{u} , \mathbf{v} \right \rangle  = 0 \quad \text{for all } \mathbf{u}, \mathbf{v} \in C$

\begin{definition}{$(n=2t, k=t)$ Quasi Cyclic Code\\}
    A systematic quasi-cyclic $(2t, t)$ code has the form $$[l(x), l(x)a(x)] \mod{x^t + 1}$$ where $i(x), a(x)\in \mathbb{F}_2[x]/x^t+1$
\end{definition}

Quasi-cyclic codes can be seen as generalizations of cyclic codes and often possess similar algebraic properties, making them easier to analyze and decode. The generator matrix of a quasi-cyclic code has a block circulant structure, which can be leveraged for efficient implementation of encoding and decoding algorithms. Quasi-cyclic codes are widely used in communication systems due to their balance of structure and performance, offering good error correction capabilities with relatively simple implementation.

\subsection{Smoothing Parameter And Learning Parity With Noise}

In code-based cryptography, researchers are interested in LPN problem and its hardness connection with other problems. 

 \begin{definition}[LPN problem]
    The (decisional) LPN problem with secret length $n$ and noise rate $\mu \in (0, 1/2)$, denoted by $LPN_{n,\mu}$, challenges to distinguish $(\mathbf{a}, \langle \mathbf{a} ,\mathbf{s}\rangle + e \mod{2})$ and $(\mathbf{a}, u)$, where $\mathbf{a} \overset{\$}{\leftarrow} \mathbf{Z}_{2}^{n}, \mathbf{s} \overset{\$}{\leftarrow} \mathbf{Z}_{2}^{n}, u \overset{\$}{\leftarrow} \mathbf{Z}_{2}$, and $e \overset{\$}{\leftarrow} Ber(\mu)$, where $Ber(\mu)$ denotes Bernoulli distribution with parameter $\mu$ and $u$ is uniform over binary field.
\end{definition}

\begin{definition}{(Linear Code) Average case Decoding Problem - aDP(n, k, t)}
        \begin{itemize}
        \item Input: $(\mathbf{G},\mathbf{y}\doteq \mathbf{xG}+\mathbf{e})$ where $\mathbf{G}$ is the random generator of an binary $[n, k]$ linear code. $\mathbf{x}\in \mathbb{F}_{2}^k, \mathbf{G}\in \mathbb{F}_{2}^{k\times n}$ and $\mathbf{e}\in \mathbb{F}_{2}^n$ with Hamming weight $t$.
        \item Output: $\mathbf{xG}$
    \end{itemize}
\end{definition}

LPN is a average case problem, and in order to finish a average case to average case reduction, i.e. reduction from LPN to aDP(n,k,t), smoothing technique plays the main role. As noted in \cite{brakerski2019worst}\cite{debris2022worst}, given a linear code with generator $\mathbf{G}\in \mathbb{F}_2^{k\times n}$, a codeword $\mathbf{c}=\mathbf{m}\mathbf{G}+\mathbf{e}$, and the noise weight $wt(\mathbf{e})=t$. We sample vector $\mathbf{r}\in \mathbb{F}_2^n$ uniformly from a noise distribution, s.t. $P(\left\langle \mathbf{r},\mathbf{e} \right\rangle=1)=p$, and then multiply it with codeword $\mathbf{c}$. Thus 
$${\left\langle \mathbf{r},\mathbf{c} \right\rangle}={\left\langle \mathbf{rG^T},\mathbf{m} \right\rangle}+{\left\langle \mathbf{r},\mathbf{e} \right\rangle}$$ which can be fed into the LPN oracle and the error is measured by statistical distance $\Delta \left((\mathbf{rG}^T, \langle \mathbf{r}, \mathbf{e}\rangle), (U_{\mathbb{F}_2^k}, \mathrm{Ber}_p)\right )$. When $\mathbf{G}$ is random over all $[n,k]$ linear code, we can achieve a average case reduction. In this paper, Rényi divergence is considered instead of statistical distance.

\subsection{Inequalities}
To help analyze smoothing bound, some inequalities will be required, especially rearrangement inequality, which can be refered to \cite{cvetkovski2012inequalities}.

\begin{lemma}[Rearrangement Inequality]
    \label{lemma:Rearrangement_inequality}
Let $a_1 \leq a_2 \leq \dots \leq a_n$ and $b_1 \leq b_2 \leq \dots \leq b_n$ be two sequences of real numbers. Then their rearrangement yields the maximum value when the sequences are ordered in the same way, that is,
\[
a_1b_1 + a_2b_2 + \dots + a_nb_n \geq a_1b_{\sigma(1)} + a_2b_{\sigma(2)} + \dots + a_nb_{\sigma(n)}
\]
for any permutation $\sigma$ of $\{1, 2, \dots, n\}$.
\end{lemma}

\begin{lemma}
 \label{lemma:AMGMlike_inequality}
    For non-negative numbers $x_1, ..., x_n$, let \( \sigma_j \) be any permutation of \( 1, 2, \ldots, n \) for any $j$, then there is a inequality as follows,
    $$
        \sum_{i=1}^{n} \prod_{j}  x_{\sigma_j(i)}^{p_j} \leq \sum_{i=1}^{n} x_i^{p}.
    $$
    Here $p_j$'s satisfy $\sum_j p_j=p$, and are fractions with integers over a common denominator no smaller than $1$, i.e. $p_j\in\frac{\mathbb{Z}}{q}, q\geq 1$ for $1\leq j\leq n$.
   
\end{lemma}
\begin{proof}
    Firstly assume $q=1$, thus $p_j$'s are all integers. Notice that by AM-GM inequality, we have 
    \begin{align*}
        \sum_{j} p_j x_{\sigma_j(i)}^p &= \sum_{j} \underbrace{x_{\sigma_j(i)}^p + \cdots + x_{\sigma_j(i)}^p}_{{p_j}} \\
        &\geq \left(\sum_{j} p_j\right)\left(\prod_{j} x_{\sigma_j(i)}^{pp_j}\right)^{\frac{1}{\sum_{j} p_j}} = p \prod_{j} x_{\sigma_j(i)}^{p_j}.
    \end{align*}
    Thus lemma can be derivd as follows,
    \begin{align*}
        \sum_{i=1}^n \prod_{j} x_{\sigma_j(i)}^{p_j} &\leq \frac{1}{p} \sum_{i=1}^n \sum_{j} p_j x_{\sigma_j(i)}^p \\
        &= \frac{1}{p} \sum_{j} p_j \sum_{i=1}^n x_{\sigma_j(i)}^p \\
        &= \frac{1}{p} \sum_{j} p_j \sum_{i=1}^n x_{i}^p = \sum_{i=1}^n x_{i}^p.
    \end{align*}
    The proof for the case where $q>1$ can also be established using a similar procedure, with a slight modification by replacing each $x_i$ with $x_i^{{1}/{q}}$.

\end{proof}

\section{Random Linear Code Smoothing}
\label{sec:Random_Linear_Code_Smoothing}

Linear codes form  a balanced set. Denote $\mathcal{B}$ as the set of all linear codes. Then we have $|\mathcal{B}|={\begin{bmatrix}n\\k\end{bmatrix}}_{q}$ where the Gaussian binomial coefficient is $${\displaystyle {\begin{bmatrix}n\\k\end{bmatrix}}_{q}}
=\frac {(1-q^{n})(1-q^{n-1})\cdots (1-q^{n-k+1})}{(1-q)(1-q^{q})\cdots (1-q^{k})}.$$
Each non-zero codeword belongs to ${\begin{bmatrix}n-1\\k-1\end{bmatrix}}_{q}$ linear codes in $\mathcal{B}$. Thus it's easy to prove the averaging lemma of linear codes as follows.

\begin{lemma}[Averaging Lemma for Linear Codes \cite{loeliger1994basic}]
    For balanced set $\mathcal{B}$ containing all linear codes with encoder $F$, and any function $f(\cdot)$, there is an identity that 
    $$\frac{1}{|\mathcal{B}|} \sum_{F\in \mathcal{B}} \sum_{\mathbf{a}\in \mathbb{F}_q^k/\{\mathbf{0}\}} f(F(\mathbf{a})) = \frac{q^k - 1}{q^n - 1} \sum_{\mathbf{c}\in \mathbb{F}_q^n/\{\mathbf{0}\}} f(\mathbf{c}).  $$
    or equivalently,
    $$\mathbb{E}_{F\sim \mathcal{B}}\sum_{\mathbf{a}\in \mathbb{F}_q^k/\{\mathbf{0}\}} f(F(\mathbf{a})) = \frac{q^k - 1}{q^n - 1} \sum_{\mathbf{c}\in \mathbb{F}_q^n/\{\mathbf{0}\}} f(\mathbf{c}).$$
\end{lemma}

\subsection{Rényi Divergence for $\alpha\in (1,2)$}
Combined lemma with the techniques in \cite{hayashi2016secure}, we will get the following theorem.
\begin{theorem}
    \label{theorem:smoothrandomlinearcode_1<Alpha<2}
    Let $W(\cdot|\cdot)$ denote the transition probability of the noisy channel $W$. Thus for $\alpha\in (1,2)$,
    when $W$ is an additive noise channel, for any rate $R$ satisfies 
    $$ R \geq 1 - \frac{H_\alpha(W)}{n}+\varepsilon,$$ where $\varepsilon>0$, then 
    $$\mathbb{E}_{F\sim \mathcal{B}} D_\alpha(U_{F}+N||U_{\mathbb{F}_q^n})\rightarrow 0$$ as $n\rightarrow \infty$. 
    
    Here we use $W_{F(\mathbf{a})}(\mathbf{y})$ to represent the probability $W(\mathbf{y} | F(\mathbf{a}))$ with the output $\mathbf{y}$ and input $F(\mathbf{a})$, where $F$ is a linear encoder contained in $\mathcal{B}$ and $\mathbf{a}$ is the message vector.  Here $H_\alpha(W)$ is the Rényi entropy of order $\alpha$, i.e. $$H_\alpha(W)\doteq \frac{1}{\alpha-1}\log\sum_{\mathbf{c}\in \mathbb{F}_q^n}  W(\mathbf{c}) ^{\alpha}.$$ 
    $U_{\mathbb{F}_q^n}$ is a uniform r.v. over the Hamming space,  $U_{F}$ is a uniform r.v. over the codewords of encoder $F$.
\end{theorem}

\begin{proof}

    Define the affine encoder $\Lambda: \mathbf{a}\rightarrow F(\mathbf{a}) + G$, where $F$ is an linear encoder in $\mathcal{B}$ and $G$ is a independent r.v. in $\mathbb{F}_q^n$. Affine encoder has the property that for $\mathbf{a}\neq \mathbf{a}'$, $\Lambda(\mathbf{a}')=F(\mathbf{a}-\mathbf{a}') + \Lambda(\mathbf{a})$. Noted that here $F(\mathbf{a}-\mathbf{a}')$ is independent with r.v. $\Lambda(\mathbf{a})$.
\begin{align*}
&q^{(\alpha-1) D_\alpha(U_{\Lambda(\mathbf{a})}+N||U_{\mathbb{F}_q^n})}=\sum_{\mathbf{y}\in \mathbb{F}_q^n}\frac{\left (\frac{1}{q^k} \sum_{\mathbf{a}\in \mathbb{F}_q^k} W_{\Lambda(\mathbf{a})}(\mathbf{y}) \right )^{\alpha}}{(\frac{1}{q^n})^{\alpha-1}}=\frac{q^{n(\alpha-1)}}{q^{k\alpha}}\sum_{y\in \mathbb{F}_q^n}  \left ( \sum_{\mathbf{a}\in \mathbb{F}_q^k} W_{\Lambda(\mathbf{a})}(\mathbf{y}) \right )^{\alpha}\\
&\mathbb{E}_\Lambda q^{(\alpha-1) D_\alpha(U_{\Lambda(\mathbf{a})}+N||U_{\mathbb{F}_q^n})}\\
&=\frac{q^{n(\alpha-1)}}{q^{k\alpha}}\sum_{\mathbf{y}\in \mathbb{F}_q^n} \mathbb{E}_\Lambda \left ( \sum_{\mathbf{a}\in \mathbb{F}_q^k} W_{\Lambda(\mathbf{a})}(\mathbf{y}) \right )^{\alpha}\\
&=\frac{q^{n(\alpha-1)}}{q^{k\alpha}}\sum_{\mathbf{y}\in \mathbb{F}_q^n} \mathbb{E}_\Lambda \sum_{\mathbf{a}\in \mathbb{F}_q^k} W_{\Lambda(\mathbf{a})}(y)\left ( \sum_{\mathbf{a}'\in \mathbb{F}_q^k} W_{\Lambda(\mathbf{a}')}(y) \right )^{\alpha-1}\\
 \intertext{ By using Jensen's inequality for $\mathbb{E}[X^{\alpha-1}] \leq (\mathbb{E}[X])^{\alpha-1}, 1 < \alpha < 2$, we derive} 
&\leq\frac{q^{n(\alpha-1)}}{q^{k\alpha}}\sum_{\mathbf{y}\in \mathbb{F}_q^n}  \sum_{\mathbf{a}\in \mathbb{F}_q^k} \mathbb{E}_{\Lambda(\mathbf{a})}W_{\Lambda(\mathbf{a})}(\mathbf{y})\left ( W_{\Lambda(\mathbf{a})}+\mathbb{E}_{F|\Lambda(\mathbf{a})}\sum_{\mathbf{a}'\in \mathbb{F}_q^k,\mathbf{a}'\neq \mathbf{a}} W_{\Lambda(\mathbf{a}')}(\mathbf{y}) \right )^{\alpha-1}\\
&=\frac{q^{n(\alpha-1)}}{q^{k\alpha}}\sum_{\mathbf{y}\in \mathbb{F}_q^n}  \sum_{\mathbf{a}\in \mathbb{F}_q^k} \mathbb{E}_{\Lambda(\mathbf{a})}W_{\Lambda(\mathbf{a})}(\mathbf{y})\left ( W_{\Lambda(\mathbf{a})}(\mathbf{y}) +\mathbb{E}_{F|\Lambda(\mathbf{a})}\sum_{\mathbf{a}'\in \mathbb{F}_q^k,\mathbf{a}'\neq \mathbf{a}} W_{\Lambda(\mathbf{a})+F(\mathbf{a}'-\mathbf{a})}(\mathbf{y}) \right )^{\alpha-1}\\
\intertext{ By using the averaging lemma, let $\mathbf{x}=\mathbf{a}'-\mathbf{a}, f(F(\mathbf{x}))=W_{\Lambda(\mathbf{a})+F(\mathbf{a}'-\mathbf{a})}(\mathbf{y})$} 
&=\frac{q^{n(\alpha-1)}}{q^{k\alpha}}\sum_{\mathbf{y}\in \mathbb{F}_q^n}  \sum_{\mathbf{a}\in \mathbb{F}_q^k} \mathbb{E}_{\Lambda(\mathbf{a})}W_{\Lambda(\mathbf{a})}(\mathbf{y})\left ( W_{\Lambda(\mathbf{a})}(\mathbf{y})  +  \frac{q^k-1}{q^{n}-1}\sum_{\mathbf{c}\in \mathbb{F}_q^n,\mathbf{c}\neq \mathbf{0}} W_{\Lambda(\mathbf{a})+\mathbf{c}}(\mathbf{y}) \right )^{\alpha-1}\\
&\leq\frac{q^{n(\alpha-1)}}{q^{k\alpha}}\sum_{\mathbf{y}\in \mathbb{F}_q^n}  \sum_{\mathbf{a}\in \mathbb{F}_q^k} \mathbb{E}_{\Lambda(\mathbf{a})}W_{\Lambda(\mathbf{a})}(\mathbf{y})\left ( W_{\Lambda(\mathbf{a})}(\mathbf{y})  +  \frac{q^k-1}{q^{n}-1} \right )^{\alpha-1}\\
\intertext{ By using $(x+y)^{\alpha-1}\leq x^{\alpha-1}+y^{\alpha-1}$ and $\frac{q^k-1}{q^{n}-1}< \frac{q^k}{q^{n}}$}
&\leq\frac{q^{n(\alpha-1)}}{q^{k\alpha}}\sum_{\mathbf{y}\in \mathbb{F}_q^n}  \sum_{\mathbf{a}\in \mathbb{F}_q^k} \mathbb{E}_{\Lambda(\mathbf{a})}W_{\Lambda(\mathbf{a})}(\mathbf{y})^{\alpha}     + \frac{q^{n(\alpha-1)}}{2^{k\alpha}}\sum_{\mathbf{y}\in \mathbb{F}_2^n}  \sum_{\mathbf{a}\in \mathbb{F}_2^k} \mathbb{E}_{\Lambda(\mathbf{a})} \frac{2^{k(\alpha-1)}}{2^{n(\alpha-1)}}W_{\Lambda(\mathbf{a})}(\mathbf{y})\\
&= q^{(\alpha-1) n(1-R-\frac{H_\alpha(W)}{n})}  + 1
\end{align*}
    The limitation goes to $1$ as $1-R-\frac{H_\alpha(W)}{n}<0$. Since channel $W$ is regular \cite{delsarte1982algebraic}, we can remove the r.v. G, and the final result proven.
    
\end{proof}

\subsection{Rényi Divergence for $\alpha\in \mathbb{N}$}
The theorem above can be extended to order $\alpha \in (2, +\infty)$. To begin with, let's first extend linear codes averaging lemma to more general case.

\begin{lemma}[Extended Averaging Lemma for Linear Codes]
\label{lemma:Extended_Averaging_LinearCodeLemma}
    For balanced set $\mathcal{B}$ containing all linear codes with encoder $F$, integer $r$, real numbers $\alpha_1, \alpha_2, ..., \alpha_r$ and any function $f(\cdot)\geq 0$, there is an inequality that 
    $$\frac{1}{|\mathcal{B}|} \sum_{F\in \mathcal{B}} \sum_{\substack{\{\mathbf{a}_1,\mathbf{a}_2,...,\mathbf{a}_r\}\\ \subseteq  \mathbb{F}_q^k / \{\mathbf{0}\}}}\prod_{i=1}^{r}f^{p_i}(F(\mathbf{a}_i)) \leq  \sum_{j} (\frac{q^k - 1}{q^n - 1})^{r-j}  \sum_{\substack{\{\mathbf{c}_1,\mathbf{c}_2,...,\mathbf{c}_r\} \subseteq \mathbb{F}_q^n / \{\mathbf{0}\}\\rank\{\mathbf{c}_1,\mathbf{c}_2,...,\mathbf{c}_r\}=r-j}}\prod_{i=1}^{r}f^{\alpha_i}(\mathbf{c}_i).$$
    or equivalently,
    $$\mathbb{E}_{F\sim \mathcal{B}} \sum_{\substack{\{\mathbf{a}_1,\mathbf{a}_2,...,\mathbf{a}_r\}\\ \subseteq  \mathbb{F}_q^k / \{\mathbf{0}\}}}\prod_{i=1}^{r}f^{\alpha_i}(F(\mathbf{a}_i)) \leq  \sum_{j} (\frac{q^k - 1}{q^n - 1})^{r-j}  \sum_{\substack{\{\mathbf{c}_1,\mathbf{c}_2,...,\mathbf{c}_r\} \subseteq \mathbb{F}_q^n / \{\mathbf{0}\}\\rank\{\mathbf{c}_1,\mathbf{c}_2,...,\mathbf{c}_r\}=r-j}}\prod_{i=1}^{r}f^{\alpha_i}(\mathbf{c}_i).$$
\end{lemma}
\begin{proof}
    Given a set of vectors $\{\mathbf{c}_1,\mathbf{c}_2,...,\mathbf{c}_r\}$ and its rank $r-j$, the span of this set forms a subspace with dimension $r-j$. Consequently, the number of linear codes containing this subspace is given by the gaussian binomial ${\begin{bmatrix}n-r+j\\k-r+j\end{bmatrix}}_{q}$. The probability of random linear codes containing the subspace is  
    $$ P\left(\begin{array}{c} \{\mathbf{c}_1,\mathbf{c}_2,...,\mathbf{c}_r\}\subseteq  F \\ \text{rank}\{\mathbf{c}_1,\mathbf{c}_2,...,\mathbf{c}_r\}=r-j \end{array}\right) = \frac{{\begin{bmatrix}n-r+j\\k-r+j\end{bmatrix}}_{q}}{{\begin{bmatrix}n\\k\end{bmatrix}}_{q}}\leq \left(\frac{q^k - 1}{q^n - 1}\right)^{r-j}. $$
    \begin{align*}
    &\mathbb{E}_{F\sim \mathcal{B}} \sum_{\substack{\{\mathbf{a}_1,\mathbf{a}_2,...,\mathbf{a}_r\}\\ \subseteq  \mathbb{F}_q^k / \{\mathbf{0}\}}}\prod_{i=1}^{r}f^{\alpha_i}(F(\mathbf{a}_i))\\
    &= \mathbb{E}_{F\sim \mathcal{B}} \sum_{j}\sum_{\substack{\{\mathbf{c}_1,\mathbf{c}_2,...,\mathbf{c}_r\} \subseteq \mathbb{F}_q^n / \{\mathbf{0}\}\\ \text{rank}\{\mathbf{c}_1,\mathbf{c}_2,...,\mathbf{c}_r\}=r-j}}\prod_{i=1}^{r}f^{\alpha_i}(\mathbf{c}_i) \mathbbm{1}\left\{\begin{array}{c} \{\mathbf{c}_1,\mathbf{c}_2,...,\mathbf{c}_r\}\subseteq  F \\ \text{rank}\{\mathbf{c}_1,\mathbf{c}_2,...,\mathbf{c}_r\}=r-j \end{array}\right\}\\
    &= \sum_{j}\sum_{\substack{\{\mathbf{c}_1,\mathbf{c}_2,...,\mathbf{c}_r\} \subseteq \mathbb{F}_q^n / \{\mathbf{0}\}\\rank\{\mathbf{c}_1,\mathbf{c}_2,...,\mathbf{c}_r\}=r-j}}\prod_{i=1}^{r}f^{\alpha_i}(\mathbf{c}_i) \mathbb{E}_{F\sim \mathcal{B}}  \mathbbm{1}\left\{\begin{array}{c} \{\mathbf{c}_1,\mathbf{c}_2,...,\mathbf{c}_r\}\subseteq  F \\ \text{rank}\{\mathbf{c}_1,\mathbf{c}_2,...,\mathbf{c}_r\}=r-j \end{array}\right\}\\
    &= \sum_{j}\sum_{\substack{\{\mathbf{c}_1,\mathbf{c}_2,...,\mathbf{c}_r\} \subseteq \mathbb{F}_q^n / \{\mathbf{0}\}\\rank\{\mathbf{c}_1,\mathbf{c}_2,...,\mathbf{c}_r\}=r-j}}\prod_{i=1}^{r}f^{\alpha_i}(\mathbf{c}_i) P\left(\begin{array}{c} \{\mathbf{c}_1,\mathbf{c}_2,...,\mathbf{c}_r\}\subseteq  F \\ \text{rank}\{\mathbf{c}_1,\mathbf{c}_2,...,\mathbf{c}_r\}=r-j \end{array}\right)\\
    &\leq \sum_{j} (\frac{q^k - 1}{q^n - 1})^{r-j}  \sum_{\substack{\{\mathbf{c}_1,\mathbf{c}_2,...,\mathbf{c}_r\} \subseteq \mathbb{F}_q^n / \{\mathbf{0}\}\\rank\{\mathbf{c}_1,\mathbf{c}_2,...,\mathbf{c}_r\}=r-j}}\prod_{i=1}^{r}f^{\alpha_i}(\mathbf{c}_i).
    \end{align*}
\end{proof}

\begin{lemma}
    \label{lemma:Each_Component_Coverge_SmoothingLemma}
    Suppose $W$ is an additive noise channel with $W(\mathbf{y}|\mathbf{x})=W(\mathbf{y}-\mathbf{x})$, and for any rate $R$ satisfies 
    $$ R \geq 1 - \frac{H_\alpha(W)}{n}+\varepsilon,$$ where $\varepsilon>0$ and  $H_\alpha(W)\doteq \frac{1}{\alpha-1}\log\sum_{ \mathbf{x}\in \mathbb{F}_q^n}  W(\mathbf{x}) ^{\alpha} $.
    For $0<j<r$, fixed $r$, denote $\{\mathbf{c}_1,\mathbf{c}_2,...,\mathbf{c}_r\}$ a arbitrary set of non-zero codewords from $\mathbb{F}_q^n$ with rank $r-j$.     
    Then for any subset $\{\mathbf{c}_{i_{j+1}},\mathbf{c}_{i_{j+2}},...,\mathbf{c}_{i_{r}}\}$ with full rank $r-j$, $1\leq i_{j+1}, i_{j+2}, ..., i_r\leq r$, and sufficient large $n$,
    $$\sum_{\mathbf{y}\in \mathbb{F}_q^n}\frac{q^{n(\alpha-1)}}{q^{k\alpha}}\left(\frac{q^k - 1}{q^n - 1}\right)^{r-j} \sum_{\substack{\{\mathbf{c}_1,\mathbf{c}_2,...,\mathbf{c}_r\} \subseteq \mathbb{F}_q^n / \{\mathbf{0}\} \\ \text{rank}\{\mathbf{c}_1,\mathbf{c}_2,...,\mathbf{c}_r\}=r-j\\ \text{rank}\{\mathbf{c}_{i_{j+1}},\mathbf{c}_{i_{j+2}},...,\mathbf{c}_{i_{r}}\}=r-j}} \prod_{i=1}^{r}W_{\mathbf{c}_i}(\mathbf{y})^{\alpha_i} \leq 
    O(q^{-\varepsilon n}) $$ where  $\alpha_i\in\frac{\mathbb{Z}}{q}$ with some $q>1$ and $\sum_{\substack{i=1}}^r\alpha_i=\alpha$.

    Furthur more, by summation over all subsets $\{\mathbf{c}_{i_{j+1}},\mathbf{c}_{i_{j+2}},...,\mathbf{c}_{i_{r}}\}$ with rank $r-j$, the following inequality is derived,
     $$\sum_{\mathbf{y}\in \mathbb{F}_q^n}\frac{q^{n(\alpha-1)}}{q^{k\alpha}}\left(\frac{q^k - 1}{q^n - 1}\right)^{r-j} \sum_{\substack{\{\mathbf{c}_1,\mathbf{c}_2,...,\mathbf{c}_r\} \subseteq \mathbb{F}_q^n / \{\mathbf{0}\} \\ \text{rank}\{\mathbf{c}_1,\mathbf{c}_2,...,\mathbf{c}_r\}=r-j}} \prod_{i=1}^{r}W_{\mathbf{c}_i}(\mathbf{y})^{\alpha_i} \leq 
    O(q^{-\varepsilon n}) $$ 
\end{lemma}
\begin{proof}
    
    Denote 
    $S=\{\mathbf{c}_1,\mathbf{c}_2,...,\mathbf{c}_r\}$, $\mathrm{rank}(S)=r-j$.     Without loss of generalization, assume $\mathbf{c}_{j+1},\mathbf{c}_{j+2},...,\mathbf{c}_r$ forms the basis of $S$, and $\mathbf{c}_1,...,\mathbf{c}_j\in span\{\mathbf{c}_{j+1},\mathbf{c}_{j+2},...,\mathbf{c}_r\}$. Consequently $\mathbf{c}_1,...,\mathbf{c}_j$ can be expressed with linear combinations of the basis as follows,
    $$\begin{aligned}
    \mathbf{c}_1 &= a_{1,j+1}\mathbf{c}_{j+1} + a_{1,j+2}\mathbf{c}_{j+2} + \ldots + a_{1,r}\mathbf{c}_r \\
    \mathbf{c}_2 &= a_{2,j+1}\mathbf{c}_{j+1} + a_{2,j+2}\mathbf{c}_{j+2} + \ldots + a_{2,r}\mathbf{c}_r \\
    &\vdots \\
    \mathbf{c}_j &= a_{j,j+1}\mathbf{c}_{j+1} + a_{j,j+2}\mathbf{c}_{j+2} + \ldots + a_{j,r}\mathbf{c}_r \\
    \end{aligned},
    $$
    or equivalently in matrix form as follows,
    $$\begin{bmatrix}
        \mathbf{c}_1 \\
        \mathbf{c}_2 \\
        \vdots \\
        \mathbf{c}_j
        \end{bmatrix}
        =
        \mathbf{A}
        \begin{bmatrix}
        \mathbf{c}_{j+1} \\
        \mathbf{c}_{j+2} \\
        \vdots \\
        \mathbf{c}_r
        \end{bmatrix}.$$
    Notice that in Hamming space the coefficients of the basis take values either $0$ to $q-1$.  
    
    To better deal with relations among different $\mathbf{c}_i's$, we will first  partition the set $S$ is  into $r-j$ disjoint subsets, labeled $C^\mathbf{A}_{j+1}, C^\mathbf{A}_{j+2}, ..., C^\mathbf{A}_{r}$, collectively exhaustive of $S$. The partitioning algorithm involves assigning $c_i$ to set $C^\mathbf{A}_i$ for all $j+1 \leq i \leq r$. For $1 \leq i \leq j$, each $c_i$ is allocated to the subset $C^\mathbf{A}_{r'}$, where $r'$ represents the highest index resulting in a non-zero value of $a_{i,r'}$. Thus by using Lemma \ref{lemma:AMGMlike_inequality} for each subset $C_{r'}^\mathbf{A}$, we derive

    \begin{align}
      \sum_{\mathbf{c}_{r'} \in \mathbb{F}_q^n} \prod_{\mathbf{c}_i \in C_{r'}^\mathbf{A}} W_{\mathbf{c}_i}(\mathbf{y})^{\alpha_i} &= \sum_{\mathbf{c}_{r'} \in \mathbb{F}_q^n} \prod_{\mathbf{c}_i \in C_{r'}^\mathbf{A}} W(\mathbf{y} - \mathbf{c}_i)^{\alpha_i}\nonumber \\
      &\leq  \sum_{\mathbf{c}_{r'} \in \mathbb{F}_q^n} W(\mathbf{y} - \mathbf{c}_i)^{\alpha_{r'}^\mathbf{A}} \nonumber\\
      &= q^{\left( 1 - \alpha_{r'}^\mathbf{A} \right) H_{\alpha_{r'}^\mathbf{A}}(W)}.  \label{ineq:CrA_to_RenyiEntropy}
    \end{align}
    Here $\alpha_{r'}^\mathbf{A}$ represents the summation of exponents $\alpha_i$'s corresponding to every $\mathbf{c}_i$ in set $C_{r'}^\mathbf{A}$. Note that  the basis vectors coefficients of $\mathbf{c}_1, ...,\mathbf{c}_j$ are fixed given $\mathbf{A}$. 

    \begin{align*}
    &\sum_{\mathbf{y}\in \mathbb{F}_q^n}\frac{q^{n(\alpha-1)}}{q^{k\alpha}}\left(\frac{q^k - 1}{q^n - 1}\right)^{r-j}  \sum_{\substack{\{\mathbf{c}_1,\mathbf{c}_2,...,\mathbf{c}_r\} \subseteq \mathbb{F}_q^n / \{\mathbf{0}\}\\rank\{\mathbf{c}_1,\mathbf{c}_2,...,\mathbf{c}_r\}=r-j\\\text{rank}\{\mathbf{c}_{{j+1}},\mathbf{c}_{{j+2}},...,\mathbf{c}_{{r}}\}=r-j}}\prod_{i=1}^{r}W_{\mathbf{c}_i}(\mathbf{y})^{\alpha_i}\\ 
    &=  \sum_{\mathbf{y}\in \mathbb{F}_q^n} \frac{q^{n(\alpha-1)}}{q^{k\alpha}}\left(\frac{q^k - 1}{q^n - 1}\right)^{r-j}  \sum_{\substack{\{\mathbf{c}_{j+1},\mathbf{c}_{j+2},...,\mathbf{c}_r\} \subseteq \mathbb{F}_q^n / \{\mathbf{0}\}\\rank\{\mathbf{c}_{j+1},\mathbf{c}_{j+2},...,\mathbf{c}_r\}=r-j  }}  \sum_{\mathbf{A}} \prod_{ \mathbf{c}_{i_{j+1}}\in C^\mathbf{A}_{j+1} }W_{\mathbf{c}_{i_{j+1}}}(\mathbf{y})^{\alpha_{i_{j+1}}}\cdots  \prod_{ \mathbf{c}_{i_{r}}\in C^\mathbf{A}_{r} }   W_{\mathbf{c}_{i_{r}}}(\mathbf{y})^{\alpha_{i_{r}}}\\
     &\leq \sum_{\mathbf{y}\in \mathbb{F}_q^n} \frac{q^{n(\alpha-1)}}{q^{k\alpha}}\left(\frac{q^k - 1}{q^n - 1}\right)^{r-j}  \sum_{\substack{\{\mathbf{c}_{j+1},\mathbf{c}_{j+2},...,\mathbf{c}_r\} \subseteq \mathbb{F}_q^n }}  \sum_{\mathbf{A}} \prod_{ \mathbf{c}_{i_{j+1}}\in C^\mathbf{A}_{j+1} }W_{\mathbf{c}_{i_{j+1}}}(\mathbf{y})^{\alpha_{i_{j+1}}}\cdots  \prod_{ \mathbf{c}_{i_{r}}\in C^\mathbf{A}_{r} }   W_{\mathbf{c}_{i_{r}}}(\mathbf{y})^{\alpha_{i_{r}}}\\
     &= \sum_{\mathbf{y}\in \mathbb{F}_q^n} \frac{q^{n(\alpha-1)}}{q^{k\alpha}}\left(\frac{q^k - 1}{q^n - 1}\right)^{r-j}    \sum_{\mathbf{A}} \sum_{\mathbf{c}_{j+1}\in\mathbb{F}_q^n }\prod_{ \mathbf{c}_{i_{j+1}}\in C^\mathbf{A}_{j+1} }W_{\mathbf{c}_{i_{j+1}}}(\mathbf{y})^{\alpha_{i_{j+1}}}\cdots  \sum_{\mathbf{c}_{r}\in\mathbb{F}_q^n }\prod_{ \mathbf{c}_{i_{r}}\in C^\mathbf{A}_{r} }   W_{\mathbf{c}_{i_{r}}}(\mathbf{y})^{\alpha_{i_{r}}}\\ \intertext{ by using inequality \eqref{ineq:CrA_to_RenyiEntropy} on  $C^\mathbf{A}_{r}$} 
     &\leq \sum_{\mathbf{y}\in \mathbb{F}_q^n} \frac{q^{n(\alpha-1)}}{q^{k\alpha}}\left(\frac{q^k - 1}{q^n - 1}\right)^{r-j}    \sum_{\mathbf{A}} \sum_{\mathbf{c}_{j+1}\in\mathbb{F}_q^n }\prod_{ \mathbf{c}_{i_{j+1}}\in C^\mathbf{A}_{j+1} }W_{\mathbf{c}_{i_{j+1}}}(\mathbf{y})^{\alpha_{i_{j+1}}}\\  & \hspace{5cm}\cdots \sum_{\mathbf{c}_{r-1}\in\mathbb{F}_q^n }\prod_{ \mathbf{c}_{i_{r-1}}\in C^\mathbf{A}_{r-1} }   W_{\mathbf{c}_{i_{r-1}}}(\mathbf{y})^{\alpha_{i_{r-1}}}\cdot q^{\left( 1 - \alpha_{r}^\mathbf{A} \right) H_{\alpha_{r}^\mathbf{A}}(W)}\\ \intertext{ by using inequality \eqref{ineq:CrA_to_RenyiEntropy} again on  $C^\mathbf{A}_{r-1}$} 
     &\leq \sum_{\mathbf{y}\in \mathbb{F}_q^n} \frac{q^{n(\alpha-1)}}{q^{k\alpha}}\left(\frac{q^k - 1}{q^n - 1}\right)^{r-j}    \sum_{\mathbf{A}} \sum_{\mathbf{c}_{j+1}\in\mathbb{F}_q^n }\prod_{ \mathbf{c}_{i_{j+1}}\in C^\mathbf{A}_{j+1} }W_{\mathbf{c}_{i_{j+1}}}(\mathbf{y})^{\alpha_{i_{j+1}}}\\  & \hspace{3cm}\cdots  \sum_{\mathbf{c}_{r-2}\in\mathbb{F}_q^n }\prod_{ \mathbf{c}_{i_{r-2}}\in C^\mathbf{A}_{r-2} }   W_{\mathbf{c}_{i_{r-2}}}(\mathbf{y})^{\alpha_{i_{r-2}}}\cdot q^{\left( 1 - \alpha_{r-1}^\mathbf{A} \right) H_{\alpha_{r-1}^\mathbf{A}}(W)}\cdot q^{\left( 1 - \alpha_{r}^\mathbf{A} \right) H_{\alpha_{r}^\mathbf{A}}(W)}\\ \intertext{ by using inequality \eqref{ineq:CrA_to_RenyiEntropy} on  $C^\mathbf{A}_{r-2}, C^\mathbf{A}_{r-3}, ..., C^\mathbf{A}_{j+1}$ separately step by step} 
     &\leq ...\leq\sum_{\mathbf{y}\in \mathbb{F}_q^n} \frac{q^{n(\alpha-1)}}{q^{k\alpha}}\left(\frac{q^k - 1}{q^n - 1}\right)^{r-j}    \sum_{\mathbf{A}}q^{\left( 1 - \alpha_{j+1}^\mathbf{A} \right) H_{\alpha_{j+1}^\mathbf{A}}(W)}\cdots  q^{\left( 1 - \alpha_{r-1}^\mathbf{A} \right) H_{\alpha_{r-1}^\mathbf{A}}(W)}\cdot q^{\left( 1 - \alpha_{r}^\mathbf{A} \right) H_{\alpha_{r}^\mathbf{A}}(W)}\\
     &\leq \sum_{\mathbf{A}}  \prod_{i=j+1}^r   q^{-n\left( 1 - \alpha_{i}^\mathbf{A} \right)\left( 1-R-\frac{1}{n}H_{\alpha_{i}^\mathbf{A}}(W)\right)}\\ 
    &\leq O(q^{-\varepsilon n}).
    \end{align*} 
    The last step is due to the fact that  summation over $\mathbf{A}$ is a finite sum for $\mathbf{A}\in \mathbb{F}_q^{j\times(r-j)}$.

    Proof for other basis with rank $r-j$ follows the same procedure. Thus by summing over all basis, we get  
     \begin{align*}
     &\sum_{\mathbf{y}\in \mathbb{F}_q^n}\frac{q^{n(\alpha-1)}}{q^{k\alpha}}\left(\frac{q^k - 1}{q^n - 1}\right)^{r-j} \sum_{\substack{\{\mathbf{c}_1,\mathbf{c}_2,...,\mathbf{c}_r\} \subseteq \mathbb{F}_q^n / \{\mathbf{0}\} \\ \text{rank}\{\mathbf{c}_1,\mathbf{c}_2,...,\mathbf{c}_r\}=r-j}} \prod_{i=1}^{r}W_{\mathbf{c}_i}(\mathbf{y})^{\alpha_i} \\
    &=\sum_{\substack{1\leq i_{j+1},i_{j+2}, ...,i_{r} \leq r}}\sum_{\mathbf{y}\in \mathbb{F}_q^n}\frac{q^{n(\alpha-1)}}{q^{k\alpha}}\left(\frac{q^k - 1}{q^n - 1}\right)^{r-j} \sum_{\substack{\{\mathbf{c}_1,\mathbf{c}_2,...,\mathbf{c}_r\} \subseteq \mathbb{F}_q^n / \{\mathbf{0}\} \\ \text{rank}\{\mathbf{c}_1,\mathbf{c}_2,...,\mathbf{c}_r\}=r-j\\ \text{rank}\{\mathbf{c}_{i_{j+1}},\mathbf{c}_{i_{j+2}},...,\mathbf{c}_{i_{r}}\}=r-j}} \prod_{i=1}^{r}W_{\mathbf{c}_i}(\mathbf{y})^{\alpha_i} \\
    &\leq  \sum_{\substack{1\leq i_{j+1},i_{j+2}, ...,i_{r} \leq r}}
        O(q^{-\varepsilon n})\\
       &=O(q^{-\varepsilon n})
    \end{align*}

\end{proof}

Now with lemmas above, Theorem \ref{theorem:smoothrandomlinearcode_1<Alpha<2} can be extended to cases for $\alpha\in \mathbb{N}$.

\begin{theorem}[Smoothing Random Linear Code for $\alpha \in \mathbb{N}$]
    \label{theorem:smoothLinearCode_Alpha_N}
    Suppose $\alpha \in \mathbb{N}$. When $W$ is an additive noise channel with $W(\mathbf{y}|\mathbf{x})=W(\mathbf{y}-\mathbf{x})$, and for any rate $R$ satisfies 
    $$ R \geq 1 - \frac{H_\alpha(W)}{n}+\varepsilon,$$ where $\varepsilon>0$ and  $H_\alpha(W)\doteq \frac{1}{\alpha-1}\log\sum_{ \mathbf{x}\in \mathbb{F}_q^n}  W(\mathbf{x}) ^{\alpha} $, then 
    $$\mathbb{E}_{F\sim \mathcal{B}} D_\alpha(U_{F}+N||U_{\mathbb{F}_q^n})\rightarrow 0$$ as $n\rightarrow \infty$. More specifically,
    $$ \mathbb{E}_{F\sim \mathcal{B}}q^{(\alpha-1) D_\alpha(U_{F(\mathbf{a})}+N||U_{\mathbb{F}_q^n})}
           \leq  O(q^{-\varepsilon n})+1.$$
\end{theorem}

\begin{proof}
    Notice that it suffices to prove  
    $$ \frac{q^{n(\alpha-1)}}{q^{k\alpha}}\sum_{y\in \mathbb{F}_q^n}   \mathbb{E}_{F\sim \mathcal{B}} \left ( \sum_{\mathbf{a}\in \mathbb{F}_q^k/\{\mathbf{0}\}} W_{F(\mathbf{a})}(\mathbf{y}) \right )^{\alpha}\leq  O(q^{-\varepsilon n})+1. $$
    Since if we prove by induction on $\alpha$, it can be computed as 
       \begin{align*}
        &\mathbb{E}_{F\sim \mathcal{B}}q^{(\alpha-1) D_\alpha(U_{\Lambda(\mathbf{a})}+N||U_{\mathbb{F}_q^n})}\nonumber\\
        &=  \frac{q^{n(\alpha-1)}}{q^{k\alpha}}\sum_{y\in \mathbb{F}_q^n}   \mathbb{E}_{F\sim \mathcal{B}} \left ( W_{\mathbf{0}}(\mathbf{y}) + \sum_{\mathbf{a}\in \mathbb{F}_q^k/\{\mathbf{0}\}} W_{F(\mathbf{a})}(\mathbf{y}) \right )^{\alpha}\\
        &=\frac{q^{n(\alpha-1)}}{q^{k\alpha}}\sum_{y\in \mathbb{F}_q^n}   \mathbb{E}_{F\sim \mathcal{B}} \sum_{\alpha_1=0}^{\alpha} \binom{\alpha}{\alpha_1} W_{\mathbf{0}}(\mathbf{y})^{\alpha_1} \left(\sum_{\mathbf{a}\in \mathbb{F}_q^k/\{\mathbf{0}\}} W_{F(\mathbf{a})}(\mathbf{y}) \right)^{\alpha-\alpha_1}\\
       &=\frac{q^{n(\alpha-1)}}{q^{k\alpha}}\sum_{y\in \mathbb{F}_q^n}   \mathbb{E}_{F\sim \mathcal{B}} \sum_{\alpha_1=0}^{\alpha} \binom{\alpha}{\alpha_1} W_{\mathbf{0}}(\mathbf{y})^{\alpha_1} \left(\sum_{\mathbf{a}\in \mathbb{F}_q^k/\{\mathbf{0}\}} W_{F(\mathbf{a})}(\mathbf{y}) \right)^{\alpha-\alpha_1}\\
       &=\frac{q^{n(\alpha-1)}}{q^{k\alpha}}\sum_{y\in \mathbb{F}_q^n}  \mathbb{E}_{F\sim \mathcal{B}}  \left(\sum_{\mathbf{a}\in \mathbb{F}_q^k/\{\mathbf{0}\}} W_{F(\mathbf{a})}(\mathbf{y}) \right)^{\alpha} \\&\quad\quad+\sum_{\alpha_1=1}^{\alpha} \binom{\alpha}{\alpha_1} \frac{q^{n(\alpha-1)}}{q^{k\alpha}}\sum_{y\in \mathbb{F}_q^n}  W_{\mathbf{0}}(\mathbf{y})^{\alpha_1}\mathbb{E}_{F\sim \mathcal{B}}  \left(\sum_{\mathbf{a}\in \mathbb{F}_q^k/\{\mathbf{0}\}} W_{F(\mathbf{a})}(\mathbf{y}) \right)^{\alpha-\alpha_1}\\ \intertext{ by induction on all $\alpha-\alpha_1$}
       &\leq \frac{q^{n(\alpha-1)}}{q^{k\alpha}}\sum_{y\in \mathbb{F}_q^n}  \mathbb{E}_{F\sim \mathcal{B}}  \left(\sum_{\mathbf{a}\in \mathbb{F}_q^k/\{\mathbf{0}\}} W_{F(\mathbf{a})}(\mathbf{y}) \right)^{\alpha} +\sum_{\alpha_1=1}^{\alpha} \binom{\alpha}{\alpha_1} \frac{q^{n(\alpha_1-1)}}{q^{k\alpha_1}}\sum_{y\in \mathbb{F}_q^n}  W_{\mathbf{0}}(\mathbf{y})^{\alpha_1}\left(O(q^{-\varepsilon n})+1 \right )\\
        &\leq \frac{q^{n(\alpha-1)}}{q^{k\alpha}}\sum_{y\in \mathbb{F}_q^n}  \mathbb{E}_{F\sim \mathcal{B}}  \left(\sum_{\mathbf{a}\in \mathbb{F}_q^k/\{\mathbf{0}\}} W_{F(\mathbf{a})}(\mathbf{y}) \right)^{\alpha} +\sum_{\alpha_1=1}^{\alpha} \binom{\alpha}{\alpha_1} \frac{1}{q^{k}} q^{n(\alpha_1-1)\left(1-R-\frac{1}{n}H_{\alpha_1}(W)\right)}\left(O(q^{-\varepsilon n})+1 \right )\\
        &\leq \frac{q^{n(\alpha-1)}}{q^{k\alpha}}\sum_{y\in \mathbb{F}_q^n}  \mathbb{E}_{F\sim \mathcal{B}}  \left(\sum_{\mathbf{a}\in \mathbb{F}_q^k/\{\mathbf{0}\}} W_{F(\mathbf{a})}(\mathbf{y}) \right)^{\alpha} +O(q^{-k}),\\
    \end{align*}
    which mainly relies on dominant type of  $ \frac{q^{n(\alpha-1)}}{q^{k\alpha}}\sum_{y\in \mathbb{F}_q^n}  \mathbb{E}_{F\sim \mathcal{B}}  \left(\sum_{\mathbf{a}\in \mathbb{F}_q^k/\{\mathbf{0}\}} W_{F(\mathbf{a})}(\mathbf{y}) \right)^{\alpha}$ when $n$ goes to infinity.
    
    Thus 
    \begin{align}
        & \frac{q^{n(\alpha-1)}}{q^{k\alpha}}\sum_{y\in \mathbb{F}_q^n}  \mathbb{E}_{F\sim \mathcal{B}}  \left(\sum_{\mathbf{a}\in \mathbb{F}_q^k/\{\mathbf{0}\}} W_{F(\mathbf{a})}(\mathbf{y}) \right)^{\alpha}\nonumber\\
     &\leq  \frac{q^{n(\alpha-1)}}{q^{k\alpha}}\sum_{y\in \mathbb{F}_q^n}    \left (\sum_{r=1}^{ \alpha }  \sum_{\alpha_1+...+\alpha_r= 
  \alpha  }\mathbb{E}_{F\sim \mathcal{B}}\sum_{\substack{\{\mathbf{a}_1,\mathbf{a}_2,...,\mathbf{a}_r\}\nonumber\\ \subseteq  \mathbb{F}_q^k / \{\mathbf{0}\}}}    \prod_{i=1}^r W_{F(\mathbf{a}_i)}(\mathbf{y}) ^{\alpha_{i}}\right )\nonumber\\
  &\leq  \frac{q^{n(\alpha-1)}}{q^{k\alpha}}\sum_{y\in \mathbb{F}_q^n}    \left (\sum_{r=1}^{ \alpha }  \sum_{\alpha_1+...+\alpha_r= 
   \alpha  }\sum_{j} (\frac{q^k - 1}{q^n - 1})^{r-j}  \sum_{\substack{\{\mathbf{c}_1,\mathbf{c}_2,...,\mathbf{c}_r\} \subseteq \mathbb{F}_q^n / \{\mathbf{0}\}\\rank\{\mathbf{c}_1,\mathbf{c}_2,...,\mathbf{c}_r\}=r-j}}\prod_{i=1}^{r}W_{\mathbf{c}_i}(\mathbf{y}) ^{\alpha_{i}}\right )\nonumber\\
   &\leq  \frac{q^{n(\alpha-1)}}{q^{k\alpha}}\sum_{y\in \mathbb{F}_q^n}    \left (\sum_{r=1}^{ \alpha }  
  \sum_{\alpha_1+...+\alpha_r= 
   \alpha  }\sum_{j>0} (\frac{q^k - 1}{q^n - 1})^{r-j}  \sum_{\substack{\{\mathbf{c}_1,\mathbf{c}_2,...,\mathbf{c}_r\} \subseteq \mathbb{F}_q^n / \{\mathbf{0}\}\\rank\{\mathbf{c}_1,\mathbf{c}_2,...,\mathbf{c}_r\}=r-j}}\prod_{i=1}^{r}W_{\mathbf{c}_i}(\mathbf{y}) ^{\alpha_{i}}\right )  \nonumber\\ &\quad\quad\quad +     \frac{q^{n(\alpha-1)}}{q^{k\alpha}}\sum_{y\in \mathbb{F}_q^n}    \left (\sum_{r=1}^{ \alpha }  
  \sum_{\alpha_1+...+\alpha_r= 
   \alpha  }(\frac{q^k - 1}{q^n - 1})^{r}  \sum_{\substack{\{\mathbf{c}_1,\mathbf{c}_2,...,\mathbf{c}_r\} \subseteq \mathbb{F}_q^n / \{\mathbf{0}\}\\rank\{\mathbf{c}_1,\mathbf{c}_2,...,\mathbf{c}_r\}=r}}\prod_{i=1}^{r}W_{\mathbf{c}_i}(\mathbf{y}) ^{\alpha_{i}}\right ) \label{ineq:sumofj=0andj>0}
    \end{align}
   Here the last step is by splitting the sum into parts where $j>0$ and parts where $j=0$. The first parts can be estimated no larger than $\sum_{r=1}^{ \alpha }  
  \sum_{\alpha_1+...+\alpha_r= 
   \alpha  }O(q^{-\varepsilon n})$ by Lemma \ref{lemma:Each_Component_Coverge_SmoothingLemma}, and the second parts can be computed as follows, 
     \begin{align*}
        & \frac{q^{n(\alpha-1)}}{q^{k\alpha}}\sum_{y\in \mathbb{F}_q^n}    \left (\sum_{r=1}^{ \alpha }  
          \sum_{\alpha_1+...+\alpha_r= 
           \alpha  }(\frac{q^k - 1}{q^n - 1})^{r}  \sum_{\substack{\{\mathbf{c}_1,\mathbf{c}_2,...,\mathbf{c}_r\} \subseteq \mathbb{F}_q^n / \{\mathbf{0}\}\\rank\{\mathbf{c}_1,\mathbf{c}_2,...,\mathbf{c}_r\}=r}}\prod_{i=1}^{r}W_{\mathbf{c}_i}(\mathbf{y}) ^{\alpha_{i}}\right ) \\
           &\leq\sum_{r=1}^{ \alpha }  
      \sum_{\alpha_1+...+\alpha_r= 
       \alpha  } \frac{q^{n(\alpha-1)}}{q^{k\alpha}}\sum_{y\in \mathbb{F}_q^n}   (\frac{q^k - 1}{q^n - 1})^{r}  \sum_{\substack{\{\mathbf{c}_1,\mathbf{c}_2,...,\mathbf{c}_r\} \subseteq \mathbb{F}_q^n }}\prod_{i=1}^{r}W_{\mathbf{c}_i}(\mathbf{y}) ^{\alpha_{i}}\\
        &\leq\sum_{r=1}^{ \alpha }  
      \sum_{\alpha_1+...+\alpha_r= 
       \alpha  } \frac{q^{n(\alpha-1)}}{q^{k\alpha}}\sum_{y\in \mathbb{F}_q^n}  (\frac{q^k - 1}{q^n - 1})^{r}  \sum_{\substack{\mathbf{c}_1 \in \mathbb{F}_q^n }}W_{\mathbf{c}_1}(\mathbf{y})^{\alpha_1} \sum_{\substack{\mathbf{c}_2 \in \mathbb{F}_q^n }}W_{\mathbf{c}_2}(\mathbf{y})^{\alpha_2}  ...\sum_{\substack{\mathbf{c}_r \in \mathbb{F}_q^n }}W_{\mathbf{c}_r}(\mathbf{y})^{\alpha_r}  \\
           &= \sum_{r=1}^{ \alpha }  
      \sum_{\alpha_1+...+\alpha_r= 
       \alpha  } \frac{q^{n(\alpha-1)}}{q^{k\alpha}}\sum_{y\in \mathbb{F}_q^n}  (\frac{q^k - 1}{q^n - 1})^{r}         
               \prod_{i=1}^rq^{(1-\alpha_i)H_{\alpha_i}(W)} \\
        &\leq \sum_{r=1}^{ \alpha }  \sum_{\alpha_1+...+\alpha_r= \alpha  }\prod_{i=1}^rq^{n(1-\alpha_i)\left(1-R-\frac{1}{n}H_{\alpha_i}(W)\right)} \\
        &\leq \sum_{r=1}^{ \alpha -1}  \sum_{\alpha_1+...+\alpha_r= \alpha  }\prod_{i=1}^rq^{n(1-\alpha_i)\left(1-R-\frac{1}{n}H_{\alpha_i}(W)\right)}  + 1\\
        &\leq \sum_{r=1}^{ \alpha -1}  \sum_{\alpha_1+...+\alpha_r= \alpha  }O(q^{-\varepsilon n}) + 1,
    \end{align*}
     where the last step is due to the fact that there exists some $\alpha_i>1$ when $r<\alpha$.

     By substituting the estimation of first and second parts above into formula \eqref{ineq:sumofj=0andj>0}, it can be obtained 
     \begin{align*}
         &\mathbb{E}_{F\sim \mathcal{B}}q^{(\alpha-1) D_\alpha(U_{\Lambda(\mathbf{a})}+N||U_{\mathbb{F}_q^n})}\\
         &\leq \sum_{r=1}^{ \alpha -1}  \sum_{\alpha_1+...+\alpha_r= \alpha  }O(q^{-\varepsilon n}) +\sum_{r=1}^{ \alpha -1}  \sum_{\alpha_1+...+\alpha_r= \alpha  }O(q^{-\varepsilon n})+1   \\
         &  =O(q^{-\varepsilon n})+1.
     \end{align*}
     The last step is by noting that the sum is finite. Thus the proof completes
\end{proof}

\begin{example}
    Take $\alpha=4$ and $q=2$ for example.
    \begin{align*}
        &\mathbb{E}_{F\sim \mathcal{B}}2^{(\alpha-1) D_\alpha(U_{\Lambda(\mathbf{a})}+N||U_{\mathbb{F}_2^n})}\nonumber\\
        &=  \frac{2^{3n}}{2^{4k}}\sum_{y\in \mathbb{F}_2^n}   \mathbb{E}_{F\sim \mathcal{B}} \left ( \sum_{\mathbf{a}\in \mathbb{F}_2^k/\{\mathbf{0} \}} W_{F(\mathbf{a})}(\mathbf{y}) \right )^{4}\\
         &=  \underbrace{\frac{2^{3n}}{2^{4k}}\sum_{y\in \mathbb{F}_2^n}   \mathbb{E}_{F\sim \mathcal{B}} \left ( \sum_{\mathbf{a}\in \mathbb{F}_2^k/\{\mathbf{0} \}} W_{F(\mathbf{a})}(\mathbf{y})^{4} \right )}_{(i)} + \underbrace{\frac{2^{3n}}{2^{4k}}\sum_{y\in \mathbb{F}_2^n}   \mathbb{E}_{F\sim \mathcal{B}} \left ( \sum_{\{\mathbf{a}_1,\mathbf{a}_2\}\in \mathbb{F}_2^k/\{\mathbf{0} \}} W_{F(\mathbf{a}_1)}(\mathbf{y})^3W_{F(\mathbf{a}_2)}(\mathbf{y}) \right )}_{(ii)}  \\ &\quad \quad +\underbrace{\frac{2^{3n}}{2^{4k}}\sum_{y\in \mathbb{F}_2^n}   \mathbb{E}_{F\sim \mathcal{B}} \left ( \sum_{\{\mathbf{a}_1,\mathbf{a}_2\}\in \mathbb{F}_2^k/\{\mathbf{0} \}} W_{F(\mathbf{a}_1)}(\mathbf{y})^2W_{F(\mathbf{a}_2)}(\mathbf{y})^2 \right )}_{(iii)} \\ &\quad \quad + \underbrace{\frac{2^{3n}}{2^{4k}}\sum_{y\in \mathbb{F}_2^n}   \mathbb{E}_{F\sim \mathcal{B}} \left ( \sum_{\{\mathbf{a}_1,\mathbf{a}_2,\mathbf{a}_3\}\in \mathbb{F}_2^k/\{\mathbf{0} \}} W_{F(\mathbf{a}_1)}(\mathbf{y})^2W_{F(\mathbf{a}_2)}(\mathbf{y})W_{F(\mathbf{a}_3)}(\mathbf{y}) \right )}_{(iv)}\\&\quad \quad + \underbrace{\frac{2^{3n}}{2^{4k}}\sum_{y\in \mathbb{F}_2^n}   \mathbb{E}_{F\sim \mathcal{B}} \left ( \sum_{\{\mathbf{a}_1,\mathbf{a}_2,\mathbf{a}_3,\mathbf{a}_4\}\in \mathbb{F}_2^k/\{\mathbf{0} \}} W_{F(\mathbf{a}_1)}(\mathbf{y})W_{F(\mathbf{a}_2)}(\mathbf{y})W_{F(\mathbf{a}_3)}(\mathbf{y})W_{F(\mathbf{a}_4)}(y) \right )}_{(v)}\\
     &\leq  \frac{2^{n(\alpha-1)}}{2^{k\alpha}}\sum_{y\in \mathbb{F}_2^n}    \left (\sum_{r=1}^{ \alpha }  \sum_{\alpha_1+...+\alpha_r= 
  \alpha  }\mathbb{E}_{F\sim \mathcal{B}}\sum_{\substack{\{\mathbf{a}_1,\mathbf{a}_2,...,\mathbf{a}_r\}\nonumber\\ \subseteq  \mathbb{F}_2^k / \{\mathbf{0}\}}}    \prod_{i=1}^r W_{F(\mathbf{a}_i)}(\mathbf{y}) ^{\alpha_{i}}\right )\\
    \end{align*}

    It is easy to prove that 
    \begin{align*}
        &(i)\leq O(2^{3n\left (1-R-\frac{1}{n}H_{4}(W)\right)}),\\ &(ii)\leq O(2^{2n\left (1-R-\frac{1}{n}H_{3}(W)\right)}),\\  &(iii)\leq O(2^{n\left (1-R-\frac{1}{n}H_{2}(W)\right)}\cdot 2^{n\left (1-R-\frac{1}{n}H_{2}(W)\right)}).
    \end{align*}

     And 
    \begin{align*}
       &(iv)\leq\\
       &\frac{2^{3n}}{2^{4k}}\sum_{y\in \mathbb{F}_2^n}  \left ( \frac{2^{k}-1}{2^n-1}\right )^3 \left (\sum_{\substack{\{\mathbf{c}_1,\mathbf{c}_2,\mathbf{c}_3\} \subseteq \mathbb{F}_2^n / \{\mathbf{0}\} \\ \text{rank}\{\mathbf{c}_1,\mathbf{c}_2,\mathbf{c}_3\}=3}} W_{\mathbf{c}_1}(\mathbf{y})^2W_{\mathbf{c}_2}(\mathbf{y})W_{\mathbf{c}_3}(\mathbf{y}) \right ) \\&+ \frac{2^{3n}}{2^{4k}}\sum_{y\in \mathbb{F}_2^n}  \left ( \frac{2^{k}-1}{2^n-1}\right )^2 \left (\sum_{\substack{\{\mathbf{c}_1,\mathbf{c}_2\} \subseteq \mathbb{F}_2^n / \{\mathbf{0}\} \\ \text{rank}\{\mathbf{c}_1,\mathbf{c}_2,\mathbf{c}_3\}=2\\\mathbf{c}_3=\mathbf{c}_2+\mathbf{c}_1}} W_{\mathbf{c}_1}(\mathbf{y})^2W_{\mathbf{c}_2}(\mathbf{y})W_{\mathbf{c}_3}(\mathbf{y}) \right ) \\
       &\leq O(2^{n\left (1-R-\frac{1}{n}H_{2}(W)\right)}) +  O(2^{n\left (1-R-\frac{1}{n}H_{2}(W)\right)} \cdot 2^{n\left (1-R-\frac{1}{n}H_{2}(W)\right)}) + O(2^{2n\left (1-R-\frac{1}{n}H_{3}(W)\right)}). 
\end{align*}  
     Similar to $(iv)$, it can be also derived that 
     \begin{align*}
         &(v)\leq O(2^{n\left (1-R-\frac{1}{n}H_{2}(W)\right)}) +  O(2^{n\left (1-R-\frac{1}{n}H_{2}(W)\right)} \cdot 2^{n\left (1-R-\frac{1}{n}H_{2}(W)\right)}) + O(2^{2n\left (1-R-\frac{1}{n}H_{3}(W)\right)}) + 1\\
     \end{align*}
     This concludes the proof.
\end{example}

\subsection{Exponents Analysis}
We need to analyze the following dominant exponent:
 $$   \min_{\substack{\sum\alpha_i=\alpha \\\alpha_i\in\mathbb{N}}}\sum (1-\alpha_i)\left(1-R-\frac{1}{n}H_{\alpha_i}(W) \right ).$$
 Here $\alpha_i$'s are non-negative integer partition of integer $\alpha$.
 
Denote $f(x)=(1-x)\left( 1 - R -  \frac{1}{n}H_{x}(W) \right)$, which is concave due to the facts that:
\begin{align*} 
  &f(x)=(1-x)\left( 1 - R -  \frac{1}{n}H_{x}(W) \right)=(1-x)(1-R)-\frac{1}{n}\log_2\sum_{i}p_i^x \\
  &f'(x) = -(1-R) - \frac{1}{n\ln2}\frac{\sum_{i}p_i^x\ln p_i}{\sum_{i}p_i^x}\\ 
  &f''(x) =  - \frac{1}{n\ln2}\frac{\left(\sum_{i}p_i^x\ln^2 p_i\right)\left(\sum_{i}p_i^x\right) -\left(\sum_{i}p_i^x\ln p_i\right)^2 }{(\sum_{i}p_i^x)^2}\leq 0 \text{ by Cauchy inequality}.\\
\end{align*}
Since $f'(-\infty)>0$ and $f'(+\infty)<0$, $f(x)$ first increase and then decrease.
Thus $$\min_{\substack{\sum\alpha_i=\alpha \\\alpha_i\in\mathbb{N}}}\sum (1-\alpha_i)\left(1-R-\frac{1}{n}H_{\alpha_i}(W) \right ) = \min \{R+\frac{1}{n}H_{2}(W)-1, (1-\alpha)\left(1-R-\frac{1}{n}H_{\alpha}(W) \right )\}.$$
Here the minimum cannot be determined. 

For example, take $p=[0.45, 0.55], \alpha=50, R=0.9,$ then we have $$R+H_2(p)-1=0.583< (1-\alpha)\left(1-R-\frac{1}{n}H_{\alpha}(W) \right )=24.99.$$

take $p=[0.2, 0.8], \alpha=3, R=0.7,$ then we have $$R+H_2(p)-1=0.085 > (1-\alpha)\left(1-R-\frac{1}{n}H_{\alpha}(W) \right )=0.053.$$

But if we use statistical distance, we will have 
$$2\Delta^2\leq D_1\leq D_2\leq ...\leq D_\alpha \rightarrow 0, $$
$$\Rightarrow \min \{R+\frac{1}{n}H_{2}(W)-1, \max_{\alpha\in \mathbb{N}} (1-\alpha)\left(1-R-\frac{1}{n}H_{\alpha}(W) \right )\}$$

\subsection{Rényi Divergence for $\alpha>2$}
Combining techniques from Theorem \ref{theorem:smoothrandomlinearcode_1<Alpha<2} and theorem \ref{theorem:smoothLinearCode_Alpha_N}, we can derive smoothing for random linear codes for scenarios $\alpha>2$.

\begin{theorem}[Smoothing Random Linear Code for $\alpha > 2$]
    Suppose $\alpha>2$, when $W$ is an additive noise channel with $W(\mathbf{y}|\mathbf{x})=W(\mathbf{y}-\mathbf{x})$, and for any rate $R$ satisfies 
    $$ R \geq 1 - \frac{H_\alpha(W)}{n}+\varepsilon,$$ where $\varepsilon>0$ and  $H_\alpha(W)\doteq \frac{1}{\alpha-1}\log\sum_{ \mathbf{x}\in \mathbb{F}_q^n}  W(\mathbf{x}) ^{\alpha} $, then 
    $$\mathbb{E}_{F\sim \mathcal{B}} D_\alpha(U_{F}+N||U_{\mathbb{F}_q^n})\rightarrow 0$$ as $n\rightarrow \infty$. More specifically,
    $$ \mathbb{E}_{F\sim \mathcal{B}}q^{(\alpha-1) D_\alpha(U_{F(\mathbf{a})}+N||U_{\mathbb{F}_q^n})}
           \leq  O(q^{-\varepsilon n})+1.$$
\end{theorem}
\begin{proof}
   Denote ${\lfloor \alpha \rfloor}$ the largest integer no larger than $\alpha$, and $\{\alpha\}$ the difference between $\alpha$ and ${\lfloor \alpha \rfloor}$, i.e.  $\alpha = {\lfloor \alpha \rfloor} + \{ \alpha\}$. 
    \begin{align*}
        &\mathbb{E}_{F\sim \mathcal{B}}q^{(\alpha-1) D_\alpha(U_{\Lambda(\mathbf{a})}+N||U_{\mathbb{F}_q^n})}\nonumber\\
        &= \frac{q^{n(\alpha-1)}}{q^{k\alpha}}\sum_{y\in \mathbb{F}_q^n}   \mathbb{E}_{F\sim \mathcal{B}} \left ( \sum_{\mathbf{a}\in \mathbb{F}_q^k} W_{F(\mathbf{a})}(\mathbf{y}) \right )^{\alpha}\\
        &= \frac{q^{n(\alpha-1)}}{q^{k\alpha}}\sum_{y\in \mathbb{F}_q^n}   \mathbb{E}_{F\sim \mathcal{B}} \left ( \sum_{\mathbf{a}\in \mathbb{F}_q^k} W_{F(\mathbf{a})}(\mathbf{y}) \right )^{\lfloor \alpha \rfloor}\left ( \sum_{\mathbf{a}\in \mathbb{F}_q^k} W_{F(\mathbf{a})}(\mathbf{y}) \right )^{\{\alpha\} }\\
        &\leq  \frac{q^{n(\alpha-1)}}{q^{k\alpha}}\sum_{y\in \mathbb{F}_q^n}  \mathbb{E}_{F\sim \mathcal{B}}  \sum_{r=1}^{ \lfloor \alpha \rfloor }  \sum_{\alpha_1+...+\alpha_r= 
      \lfloor \alpha \rfloor  }\sum_{\substack{\{\mathbf{a}_1,\mathbf{a}_2,...,\mathbf{a}_r\}\nonumber\\ \subseteq  \mathbb{F}_q^k / \{\mathbf{0}\}}}    \prod_{i=1}^r W_{F(\mathbf{a}_i)}(\mathbf{y}) ^{\alpha_{i}}\\ &\hspace{3cm}\cdot\left ( \sum_{\mathbf{a}\in \text{span}\{\mathbf{a}_1,\mathbf{a}_2,...,\mathbf{a}_r\}} W_{F(\mathbf{a})}(\mathbf{y}) + \sum_{\mathbf{a}\notin  \text{span}\{\mathbf{a}_1,\mathbf{a}_2,...,\mathbf{a}_r\}} W_{F(\mathbf{a})}(\mathbf{y}) \right )^{\{\alpha\}} \\ \intertext{ By using $(\sum x_i)^{\{\alpha\}}\leq \sum x_i^{\{\alpha\}}$ since $\{\alpha\}\leq 1$ } 
      &\leq  \frac{q^{n(\alpha-1)}}{q^{k\alpha}}\sum_{y\in \mathbb{F}_q^n}  \mathbb{E}_{F\sim \mathcal{B}}  \sum_{r=1}^{ \lfloor \alpha \rfloor }  \sum_{\alpha_1+...+\alpha_r= 
      \lfloor \alpha \rfloor  }\sum_{\substack{\{\mathbf{a}_1,\mathbf{a}_2,...,\mathbf{a}_r\}\nonumber\\ \subseteq  \mathbb{F}_q^k / \{\mathbf{0}\}}}    \prod_{i=1}^r W_{F(\mathbf{a}_i)}(\mathbf{y}) ^{\alpha_{i}}\\ &\hspace{3cm}\cdot \left( \sum_{\mathbf{a}\in \text{span}\{\mathbf{a}_1,\mathbf{a}_2,...,\mathbf{a}_r\}} W_{F(\mathbf{a})}(\mathbf{y})^{\{\alpha\}} + \left (\sum_{\mathbf{a}\notin  \text{span}\{\mathbf{a}_1,\mathbf{a}_2,...,\mathbf{a}_r\}} W_{F(\mathbf{a})}(\mathbf{y}) \right )^{\{\alpha\}} \right )\\
      &\leq  \underbrace{\frac{q^{n(\alpha-1)}}{q^{k\alpha}}\sum_{y\in \mathbb{F}_q^n}  \mathbb{E}_{F\sim \mathcal{B}}  \sum_{r=1}^{ \lfloor \alpha \rfloor }  \sum_{\alpha_1+...+\alpha_r= 
      \lfloor \alpha \rfloor  }\sum_{\substack{\{\mathbf{a}_1,\mathbf{a}_2,...,\mathbf{a}_r\}\nonumber\\ \subseteq  \mathbb{F}_q^k / \{\mathbf{0}\}}}    \prod_{i=1}^r W_{F(\mathbf{a}_i)}(\mathbf{y}) ^{\alpha_{i}}\sum_{\mathbf{a}\in \text{span}\{\mathbf{a}_1,\mathbf{a}_2,...,\mathbf{a}_r\}} W_{F(\mathbf{a})}(\mathbf{y})^{\{\alpha\}}}_{(i)} \\ & +  \underbrace{\frac{q^{n(\alpha-1)}}{q^{k\alpha}}\sum_{y\in \mathbb{F}_q^n}  \mathbb{E}_{F\sim \mathcal{B}}  \sum_{r=1}^{ \lfloor \alpha \rfloor }  \sum_{\alpha_1+...+\alpha_r= 
      \lfloor \alpha \rfloor  }\sum_{\substack{\{\mathbf{a}_1,\mathbf{a}_2,...,\mathbf{a}_r\}\nonumber\\ \subseteq  \mathbb{F}_q^k / \{\mathbf{0}\}}}    \prod_{i=1}^r W_{F(\mathbf{a}_i)}(\mathbf{y}) ^{\alpha_{i}}\left (\sum_{\mathbf{a}\notin  \text{span}\{\mathbf{a}_1,\mathbf{a}_2,...,\mathbf{a}_r\}} W_{F(\mathbf{a})}(\mathbf{y}) \right )^{\{\alpha\}}}_{(ii)}. \\
    \end{align*}
\end{proof}
Now let's compute $(i)$ and $(ii)$ separately.

In $(i)$, $\mathbf{a}$ can be expressed by linear combination of $\{\mathbf{a}_1,\mathbf{a}_2,...,\mathbf{a}_r\}$, and thus there exists some $\mathbf{a}_j$ s.t. its corresponding coefficients is non zero. Thus by Rearangement Inequality, it can be derived
$$ \sum_{\mathbf{a}_j}W_{F(\mathbf{a}_j)}(\mathbf{y})^{\alpha_{j}}W_{F(\mathbf{a})}(\mathbf{y}) ^{\{\alpha\}}\leq  \sum_{\mathbf{a}_j}W_{F(\mathbf{a}_j)}(\mathbf{y})^{\alpha_{j}+ \{\alpha\}}.$$ And $(i)$ can be computed as  
   \begin{align*}
     (i)&\leq \sum_{r=1}^{ \lfloor \alpha \rfloor }  \sum_{\alpha_1+...+\alpha_r 
      =\lfloor \alpha \rfloor  }\sum_{\mathbf{a}} \frac{q^{n(\alpha-1)}}{q^{k\alpha}}\sum_{y\in \mathbb{F}_q^n}\mathbb{E}_{F\sim \mathcal{B}} \sum_{\substack{\{\mathbf{a}_1,\mathbf{a}_2,...,\mathbf{a}_r\}\nonumber\\ \subseteq  \mathbb{F}_q^k / \{\mathbf{0}\}}}    \prod_{i=1}^r W_{F(\mathbf{a}_i)}(\mathbf{y}) ^{\alpha_{i}'}\\
      &\leq \sum_{r=1}^{ \lfloor \alpha \rfloor }  \sum_{\alpha_1+...+\alpha_r=
      \lfloor \alpha \rfloor  }\sum_{\mathbf{a}} O(q^{-\varepsilon n})\\
      &= O(q^{-\varepsilon n}).
    \end{align*}
    In the first line $\alpha_i'$'s satisfy $\sum_{i=1}^r\alpha_i'=\alpha$, and the second line is achieved by similar techniques from proof in Theorem \ref{theorem:smoothLinearCode_Alpha_N}.

As for $(ii)$, averaging lemma of linear subspaces $(\text{span}\{\mathbf{a}_1, \ldots, \mathbf{a}_r\})^\perp
$ needs to be considered for random linear code encoder $F|F(\mathbf{a}_1), ...,F(\mathbf{a}_r)$. It can be easily derived that all linear codes that contain subspace $\text{span}\{\mathbf{a}_1, \ldots, \mathbf{a}_r\}$ is uniform over all subspaces of $(\text{span}\{\mathbf{a}_1, \ldots, \mathbf{a}_r\})^\perp$. Thus it can be derived that,
\begin{align*}
   &  \mathbb{E}_{F|F(\mathbf{a}_1), ...,F(\mathbf{a}_r)\sim \mathcal{B}'} \left (\sum_{\mathbf{a}\notin  \text{span}\{\mathbf{a}_1,\mathbf{a}_2,...,\mathbf{a}_r\}} W_{F(\mathbf{a})}(\mathbf{y}) \right )^{\{\alpha\}}\\
   &\leq  \left (\mathbb{E}_{F|F(\mathbf{a}_1), ...,F(\mathbf{a}_r)\sim \mathcal{B}'}\sum_{\mathbf{a}\notin  \text{span}\{\mathbf{a}_1,\mathbf{a}_2,...,\mathbf{a}_r\}} W_{F(\mathbf{a})}(\mathbf{y}) \right )^{\{\alpha\}}\\
    &= \left (\frac{q^{k-r}-1}{q^{n-r}-1}\sum_{\mathbf{c}\notin  \text{span}\{F(\mathbf{a}_1),F(\mathbf{a}_2),...,F(\mathbf{a}_r)\}} W_{\mathbf{c}}(\mathbf{y})\right )^{\{\alpha\}}\\
   &\leq  \left (\frac{q^{k-r}-1}{q^{n-r}-1}\right)^{\{\alpha\}}\leq \frac{q^{k\{\alpha\}}}{q^{n\{\alpha\}}}.
\end{align*}
With the subspaces averaging lemma, 
\begin{align*}
   &  (ii)\\
   &\leq  \frac{q^{n(\alpha-1)}}{q^{k\alpha}}\sum_{y\in \mathbb{F}_q^n}    \sum_{r=1}^{ \lfloor \alpha \rfloor }  \sum_{\alpha_1+...+\alpha_r= 
   \lfloor \alpha \rfloor  }\sum_{j} (\frac{q^k - 1}{q^n - 1})^{r-j}\cdot\\ &\vspace{3cm}  \sum_{\substack{\{\mathbf{c}_1,\mathbf{c}_2,...,\mathbf{c}_r\} \subseteq \mathbb{F}_q^n / \{\mathbf{0}\}\\rank\{\mathbf{c}_1,\mathbf{c}_2,...,\mathbf{c}_r\}=r-j}}\prod_{i=1}^{r}W_{\mathbf{c}_i}(\mathbf{y}) ^{\alpha_{i}}\mathbb{E}_{F|\mathbf{c}_1, ...,\mathbf{c}_r\sim \mathcal{B}'} \left (\sum_{\mathbf{a}\notin  \text{span}\{\mathbf{a}_1,\mathbf{a}_2,...,\mathbf{a}_r\}} W_{F(\mathbf{a})}(\mathbf{y}) \right )^{\{\alpha\}}\\ \intertext{ by using spacespace averaging lemma }
   &\leq    \frac{q^{n(\alpha-1)}}{q^{k\alpha}}\sum_{y\in \mathbb{F}_q^n}    \sum_{r=1}^{ \lfloor \alpha \rfloor }  \sum_{\alpha_1+...+\alpha_r= 
   \lfloor \alpha \rfloor  }\sum_{j} (\frac{q^k - 1}{q^n - 1})^{r-j}  \sum_{\substack{\{\mathbf{c}_1,\mathbf{c}_2,...,\mathbf{c}_r\} \subseteq \mathbb{F}_q^n / \{\mathbf{0}\}\\rank\{\mathbf{c}_1,\mathbf{c}_2,...,\mathbf{c}_r\}=r-j}}\prod_{i=1}^{r}W_{\mathbf{c}_i}(\mathbf{y}) ^{\alpha_{i}}\frac{q^{k\{\alpha\}}}{q^{n\{\alpha\}}} \\ 
   &=\frac{q^{n(\lfloor \alpha \rfloor-1)}}{q^{k\lfloor \alpha \rfloor}}\sum_{y\in \mathbb{F}_q^n}    \sum_{r=1}^{ \lfloor \alpha \rfloor }  \sum_{\alpha_1+...+\alpha_r= 
   \lfloor \alpha \rfloor  }\sum_{j} (\frac{q^k - 1}{q^n - 1})^{r-j}  \sum_{\substack{\{\mathbf{c}_1,\mathbf{c}_2,...,\mathbf{c}_r\} \subseteq \mathbb{F}_q^n / \{\mathbf{0}\}\\rank\{\mathbf{c}_1,\mathbf{c}_2,...,\mathbf{c}_r\}=r-j}}\prod_{i=1}^{r}W_{\mathbf{c}_i}(\mathbf{y}) ^{\alpha_{i}} \\ \intertext{ By Theorem \ref{theorem:smoothLinearCode_Alpha_N}, we derive}
   &  \leq O(q^{-\varepsilon n})+1.
\end{align*}

\section{Random Self-Dual Code Smoothing}
\label{sec:Random_Self_Dual_Code_Smoothing}

Note that self dual code has rate $R=0.5$, and is sure to have whole one vector $\mathbf{1}$. Another property is that self dual codes have been proved to be almost balanced in \cite[Corollary~2.3, 2.4]{macwilliams1972good}. 

\begin{theorem}[\cite{macwilliams1972good}]
     Let $\mathcal{B}$ denote the set of $(n=2t, k=t)$ self-dual codes in which the weight of every codeword is divisible by $4$.
    
    The number of codes in $\mathcal{B}$ is
    \[
    (2^{t -2} + 1)(2^{t - 3} + 1) \dots (2 + 1)2.
    \]
    Recall that they all contain the vectors $\mathbf{0}$, $\mathbf{1}$.

    Let $\mathbf{v}$ be a vector other than 0, 1 with $w(\mathbf{v}) \equiv 0 \pmod{4}$. The number of codes in $\mathcal{B}$ which contain $\mathbf{v}$ is
    \[
    (2^{t -3} + 1)(2^{t - 4} + 1) \dots (2 + 1)2.
    \]
\end{theorem}
Thus the averaging lemma for almost balanced set of dual codes become as follows.

\begin{lemma}[Averaging Lemma for Self-Dual Codes]
    For balanced set $\mathcal{B}$ containing encoders of $(n=2t, k=t)$ self-dual codes in which the weight of every codeword is divisible by $4$, and any function $f(\cdot)$, there is an identity that 
    $$\frac{1}{|\mathcal{B}|} \sum_{F\in \mathcal{B}} \sum_{\mathbf{a}\in \mathbb{F}_2^k/\{\mathbf{0}\}} f(F(\mathbf{a})) = \frac{1}{2^{t-2} + 1} \sum_{\substack{w(\mathbf{v}) \equiv 0 \pmod{4}\\ \mathbf{v}\neq 0, \mathbf{v}\neq 1}}  f(\mathbf{v}).   $$
    Or equivalently,
    $$ \mathbb{E}_{F\sim \mathcal{B}}\sum_{\mathbf{a}\in \mathbb{F}_2^k/\{\mathbf{0}, \mathbf{1}^{-1}\}} f(F(\mathbf{a})) = \frac{1}{2^{t-2} + 1} \sum_{\substack{w(\mathbf{v}) \equiv 0 \pmod{4}\\ \mathbf{v}\neq 0, \mathbf{v}\neq 1}}  f(\mathbf{v}).  $$
\end{lemma}

\begin{theorem}[Smoothing Random Self-Dual Code for KL Divergence]
    when $W$ is an additive noise channel, and rate $R$ satisfies 
    $$ R=0.5 > 1 - \frac{H(W)}{n}, $$ then 
    $$\mathbb{E}_{F\sim \mathcal{B}} D(U_{F}+N||U_{\mathbb{F}_2^n})\rightarrow 0$$ as $n\rightarrow \infty$. 
\end{theorem}

\begin{proof}

    Define the affine encoder $\Lambda: \mathbf{a}\rightarrow F(\mathbf{a}) + G$, where $F$ is an linear encoder in $\mathcal{B}$ and $G$ is a independent r.v. in $\mathbb{F}_2^n$. Affine encoder has the property that for $\mathbf{a}\neq \mathbf{a}'$, $\Lambda(\mathbf{a}')=F(\mathbf{a}-\mathbf{a}') + \Lambda(\mathbf{a})$. Noted that here $F(\mathbf{a}-\mathbf{a}')$ is independent with r.v. $\Lambda(\mathbf{a})$.
\begin{align*}
&2^{(\alpha-1) D_\alpha(U_{\Lambda(\mathbf{a})}+N||U_{\mathbb{F}_2^n})}=\sum_{\mathbf{y}\in \mathbb{F}_2^n}\frac{\left (\frac{1}{2^k} \sum_{\mathbf{a}\in \mathbb{F}_2^k} W_{\Lambda(\mathbf{a})}(\mathbf{y}) \right )^{\alpha}}{(\frac{1}{2^n})^{\alpha-1}}=\frac{2^{n(\alpha-1)}}{2^{k\alpha}}\sum_{y\in \mathbb{F}_2^n}  \left ( \sum_{\mathbf{a}\in \mathbb{F}_2^k} W_{\Lambda(\mathbf{a})}(\mathbf{y}) \right )^{\alpha}\\
&\mathbb{E}_\Lambda 2^{(\alpha-1) D_\alpha(U_{\Lambda(\mathbf{a})}+N||U_{\mathbb{F}_2^n})}\\
&=\frac{2^{n(\alpha-1)}}{2^{k\alpha}}\sum_{\mathbf{y}\in \mathbb{F}_2^n} \mathbb{E}_\Lambda \left ( \sum_{\mathbf{a}\in \mathbb{F}_2^k} W_{\Lambda(\mathbf{a})}(\mathbf{y}) \right )^{\alpha}\\
&=\frac{2^{n(\alpha-1)}}{2^{k\alpha}}\sum_{\mathbf{y}\in \mathbb{F}_2^n} \mathbb{E}_\Lambda \sum_{\mathbf{a}\in \mathbb{F}_2^k} W_{\Lambda(\mathbf{a})}(y)\left ( \sum_{\mathbf{a}'\in \mathbb{F}_2^k} W_{\Lambda(\mathbf{a}')}(y) \right )^{\alpha-1}\\
 \intertext{ By using Jensen's inequality for $\mathbb{E}[X^{\alpha-1}] \leq (\mathbb{E}[X])^{\alpha-1}, 1 < \alpha < 2$, we derive} 
&\leq\frac{2^{n(\alpha-1)}}{2^{k\alpha}}\sum_{\mathbf{y}\in \mathbb{F}_2^n}  \sum_{\mathbf{a}\in \mathbb{F}_2^k} \mathbb{E}_{\Lambda(\mathbf{a})}W_{\Lambda(\mathbf{a})}(\mathbf{y})\left ( W_{\Lambda(\mathbf{a})}+W_{\Lambda(\mathbf{a})+\mathbf{1}}(\mathbf{y})+\mathbb{E}_{F|\Lambda(\mathbf{a})} \sum_{\substack{\mathbf{a}'\in \mathbb{F}_2^k\\  F(\mathbf{a}'-\mathbf{a})\neq \mathbf{0}, \mathbf{1}}}   W_{\Lambda(\mathbf{a}')}(\mathbf{y}) \right )^{\alpha-1}\\
&=\frac{2^{n(\alpha-1)}}{2^{k\alpha}}\sum_{\mathbf{y}\in \mathbb{F}_2^n}  \sum_{\mathbf{a}\in \mathbb{F}_2^k} \mathbb{E}_{\Lambda(\mathbf{a})}W_{\Lambda(\mathbf{a})}(\mathbf{y})\left ( W_{\Lambda(\mathbf{a})}(\mathbf{y})+W_{\Lambda(\mathbf{a})+\mathbf{1}}(\mathbf{y}) +\mathbb{E}_{F|\Lambda(\mathbf{a})} \sum_{\substack{\mathbf{a}'\in \mathbb{F}_2^k\\  F(\mathbf{a}'-\mathbf{a})\neq \mathbf{0}, \mathbf{1}}}  W_{\Lambda(\mathbf{a})+F(\mathbf{a}'-\mathbf{a})}(\mathbf{y}) \right )^{\alpha-1}\\
\intertext{ By using the averaging lemma, let $\mathbf{x}=\mathbf{a}'-\mathbf{a}, f(F(\mathbf{x}))=W_{\Lambda(\mathbf{a})+F(\mathbf{a}'-\mathbf{a})}(\mathbf{y})$} 
&=\frac{2^{n(\alpha-1)}}{2^{k\alpha}}\sum_{\mathbf{y}\in \mathbb{F}_2^n}  \sum_{\mathbf{a}\in \mathbb{F}_2^k} \mathbb{E}_{\Lambda(\mathbf{a})}W_{\Lambda(\mathbf{a})}(\mathbf{y})\left ( W_{\Lambda(\mathbf{a})}(\mathbf{y}) + W_{\Lambda(\mathbf{a})+\mathbf{1}}(\mathbf{y})  +   \frac{1}{2^{t-2} + 1} \sum_{\substack{w(\mathbf{c}) \equiv 0 \pmod{4}\\ \mathbf{c}\neq 0, \mathbf{c}\neq 1}}  W_{\Lambda(\mathbf{a})+\mathbf{c}}(\mathbf{y})  \right )^{\alpha-1}\\
&\leq\frac{2^{n(\alpha-1)}}{2^{k\alpha}}\sum_{\mathbf{y}\in \mathbb{F}_2^n}  \sum_{\mathbf{a}\in \mathbb{F}_2^k} \mathbb{E}_{\Lambda(\mathbf{a})}W_{\Lambda(\mathbf{a})}(\mathbf{y})\left ( W_{\Lambda(\mathbf{a})}(\mathbf{y}) + W_{\Lambda(\mathbf{a})+\mathbf{1}}(\mathbf{y})  +  \frac{1}{2^{t-2} + 1}  \right )^{\alpha-1}\\
\intertext{ By using $(x+y+z)^{\alpha-1}\leq x^{\alpha-1}+y^{\alpha-1}+z^{\alpha-1}$ }
&\leq\frac{2^{n(\alpha-1)}}{2^{k\alpha}}\sum_{\mathbf{y}\in \mathbb{F}_2^n}  \sum_{\mathbf{a}\in \mathbb{F}_2^k} \mathbb{E}_{\Lambda(\mathbf{a})}W_{\Lambda(\mathbf{a})}(\mathbf{y})^{\alpha} + \frac{2^{n(\alpha-1)}}{2^{k\alpha}}\sum_{\mathbf{y}\in \mathbb{F}_2^n}  \sum_{\mathbf{a}\in \mathbb{F}_2^k} \mathbb{E}_{\Lambda(\mathbf{a})}W_{\Lambda(\mathbf{a})}(\mathbf{y})W_{\Lambda(\mathbf{a})+\mathbf{1}}(\mathbf{y})^{\alpha-1}      \\ &\hspace{2cm}+\frac{2^{n(\alpha-1)}}{2^{k\alpha}}\sum_{\mathbf{y}\in \mathbb{F}_2^n}  \sum_{\mathbf{a}\in \mathbb{F}_2^k} \mathbb{E}_{\Lambda(\mathbf{a})} (\frac{1}{2^{t-2} + 1})^{\alpha-1}W_{\Lambda(\mathbf{a})}(\mathbf{y})\\
&= 2^{2(\alpha-1) t(1-R-\frac{H_\alpha(W)}{n})} + 2^{2(\alpha-1)t(1-R-\frac{H'_\alpha(W)}{n})}  + \left(\frac{2^{2t}}{2^t(2^{t-2}+1)}\right)^{\alpha-1}
\end{align*}
    Here $$H(W)=\lim_{\alpha\rightarrow 1}H_\alpha(W),H'(W)=\lim_{\alpha\rightarrow 1}H'_\alpha(W),$$ 
    and 
    $$H'_\alpha(W)=\frac{1}{1-\alpha}\log\frac{1}{2^n}\sum_{\mathbf{y}\in \mathbb{F}_2^n, \mathbf{x}\in \mathbb{F}_2^n}  W(\mathbf{y}|\mathbf{x}) W(\mathbf{y}|\mathbf{x+1}) ^{\alpha-1}.$$ The limitation goes to $1$ when $1-R-\frac{H_\alpha(W)}{n}<0$ if we first let $t\rightarrow\infty$ and then $\alpha\rightarrow 1$. Notice that $1-R-\frac{H'_\alpha(W)}{n}<0$ due to the rearangement inequality that $H_\alpha(W)\leq H_\alpha(W')$.
    Since channel $W$ is regular \cite{delsarte1982algebraic}, we can remove the r.v. G, and the final result proven.
    
\end{proof}

\section{Random Quasi Cyclic Code Smoothing}
\label{sec:Random_Quasi_Cyclic_Code_Smoothing}

Denote $\mathcal{B}$ as the set of $(2t, t)$ quasi cyclic code with odd weight $a(x)$. It has been proven in \cite{Chen_Peterson_Weldon} that 
$|\mathcal{B}|=2^{t-1}-1$. And when $t\equiv \pm 3 \mod{8}$, and $t$ is a prime for which 2 is primitive, then each non-zero vector except all-one vector, belongs to exactly one code in $\mathcal{B}$. 
\begin{lemma}[Averaging Lemma for $(2t ,t)$ Quasi-Cyclic Codes]
    For balanced set $\mathcal{B}$ containing encoders of $(n=2t, k=t)$ quasi-cyclic codes in which $a(x)$ has odd weight, $t\equiv \pm 3 \mod{8}$, and $t$ is a prime for which 2 is primitive, and any function $f(\cdot)$, there is an identity that 
    $$\frac{1}{|\mathcal{B}|} \sum_{F\in \mathcal{B}} \sum_{\mathbf{a}\in \mathbb{F}_2^k/\{\mathbf{0, 1}\}} f(F(\mathbf{a})) =\frac{1}{2^{t-1} - 1} \sum_{\substack{w(\mathbf{v}) \equiv 0 \pmod{2}\\ \mathbf{v}\neq 0, \mathbf{v}\neq 1}}  f(\mathbf{v})  $$
    Or equivalently,
    $$ \mathbb{E}_{F\sim \mathcal{B}}\sum_{\mathbf{a}\in \mathbb{F}_2^k/\{\mathbf{0}, \mathbf{1}^{-1}\}} f(F(\mathbf{a})) = \frac{1}{2^{t-1} - 1} \sum_{\substack{w(\mathbf{v}) \equiv 0 \pmod{2}\\ \mathbf{v}\neq 0, \mathbf{v}\neq 1}}  f(\mathbf{v}).  $$
\end{lemma}

Under the same condition, we can prove smoothing of quasi cyclic codes for KL divergence.

\begin{theorem}[Smoothing Random Quasi Cyclic Code for KL Divergence]
    when $W$ is  an additive noise channel, and rate $R$ satisfies 
    $$  R=0.5 > 1 - \frac{H(W)}{n}, $$ then 
    $$\mathbb{E}_{F\sim \mathcal{B}} D(U_{F}+N||U_{\mathbb{F}_2^n})\rightarrow 0$$ as $n\rightarrow \infty$. 

    If we furthur require on $W$ that $\sum_{w(\mathbf{c}) \equiv 0 \pmod{2}}  W_{\mathbf{c}}(\mathbf{y}) = \sum_{w(\mathbf{c}) \equiv 1 \pmod{2}}  W_{\mathbf{c}}(\mathbf{y}) = \frac{1}{2}$, then for $\alpha\in (1,2)$, as $n\rightarrow \infty$, $$\mathbb{E}_{F\sim \mathcal{B}} D_\alpha(U_{F}+N||U_{\mathbb{F}_2^n})\rightarrow 0.$$
   
\end{theorem}
\begin{proof}

    Define the affine encoder $\Lambda: \mathbf{a}\rightarrow F(\mathbf{a}) + G$, where $F$ is an linear encoder in $\mathcal{B}$ and $G$ is a independent r.v. in $\mathbb{F}_2^n$. Affine encoder has the property that for $\mathbf{a}\neq \mathbf{a}'$, $\Lambda(\mathbf{a}')=F(\mathbf{a}-\mathbf{a}') + \Lambda(\mathbf{a})$. Noted that here $F(\mathbf{a}-\mathbf{a}')$ is independent with r.v. $\Lambda(\mathbf{a})$.
\begin{align*}
&2^{(\alpha-1) D_\alpha(U_{\Lambda(\mathbf{a})}+N||U_{\mathbb{F}_2^n})}=\sum_{\mathbf{y}\in \mathbb{F}_2^n}\frac{\left (\frac{1}{2^k} \sum_{\mathbf{a}\in \mathbb{F}_2^k} W_{\Lambda(\mathbf{a})}(\mathbf{y}) \right )^{\alpha}}{(\frac{1}{2^n})^{\alpha-1}}=\frac{2^{n(\alpha-1)}}{2^{k\alpha}}\sum_{y\in \mathbb{F}_2^n}  \left ( \sum_{\mathbf{a}\in \mathbb{F}_2^k} W_{\Lambda(\mathbf{a})}(\mathbf{y}) \right )^{\alpha}\\
&\mathbb{E}_\Lambda 2^{(\alpha-1) D_\alpha(U_{\Lambda(\mathbf{a})}+N||U_{\mathbb{F}_2^n})}\\
&=\frac{2^{n(\alpha-1)}}{2^{k\alpha}}\sum_{\mathbf{y}\in \mathbb{F}_2^n} \mathbb{E}_\Lambda \left ( \sum_{\mathbf{a}\in \mathbb{F}_2^k} W_{\Lambda(\mathbf{a})}(\mathbf{y}) \right )^{\alpha}\\
&=\frac{2^{n(\alpha-1)}}{2^{k\alpha}}\sum_{\mathbf{y}\in \mathbb{F}_2^n} \mathbb{E}_\Lambda \sum_{\mathbf{a}\in \mathbb{F}_2^k} W_{\Lambda(\mathbf{a})}(y)\left ( \sum_{\mathbf{a}'\in \mathbb{F}_2^k} W_{\Lambda(\mathbf{a}')}(y) \right )^{\alpha-1}\\
 \intertext{ By using Jensen's inequality for $\mathbb{E}[X^{\alpha-1}] \leq (\mathbb{E}[X])^{\alpha-1}, 1 < \alpha < 2$, we derive} 
&\leq\frac{2^{n(\alpha-1)}}{2^{k\alpha}}\sum_{\mathbf{y}\in \mathbb{F}_2^n}  \sum_{\mathbf{a}\in \mathbb{F}_2^k} \mathbb{E}_{\Lambda(\mathbf{a})}W_{\Lambda(\mathbf{a})}(\mathbf{y})\left ( W_{\Lambda(\mathbf{a})}+W_{\Lambda(\mathbf{a})+\mathbf{1}}(\mathbf{y})+\mathbb{E}_{F|\Lambda(\mathbf{a})} \sum_{\substack{\mathbf{a}'\in \mathbb{F}_2^k\\  F(\mathbf{a}'-\mathbf{a})\neq \mathbf{0}, \mathbf{1}}}   W_{\Lambda(\mathbf{a}')}(\mathbf{y}) \right )^{\alpha-1}\\
&=\frac{2^{n(\alpha-1)}}{2^{k\alpha}}\sum_{\mathbf{y}\in \mathbb{F}_2^n}  \sum_{\mathbf{a}\in \mathbb{F}_2^k} \mathbb{E}_{\Lambda(\mathbf{a})}W_{\Lambda(\mathbf{a})}(\mathbf{y})\left ( W_{\Lambda(\mathbf{a})}(\mathbf{y})+W_{\Lambda(\mathbf{a})+\mathbf{1}}(\mathbf{y}) +\mathbb{E}_{F|\Lambda(\mathbf{a})} \sum_{\substack{\mathbf{a}'\in \mathbb{F}_2^k\\  F(\mathbf{a}'-\mathbf{a})\neq \mathbf{0}, \mathbf{1}}}  W_{\Lambda(\mathbf{a})+F(\mathbf{a}'-\mathbf{a})}(\mathbf{y}) \right )^{\alpha-1}\\
\intertext{ By using the averaging lemma, let $\mathbf{x}=\mathbf{a}'-\mathbf{a}, f(F(\mathbf{x}))=W_{\Lambda(\mathbf{a})+F(\mathbf{a}'-\mathbf{a})}(\mathbf{y})$} 
&=\frac{2^{n(\alpha-1)}}{2^{k\alpha}}\sum_{\mathbf{y}\in \mathbb{F}_2^n}  \sum_{\mathbf{a}\in \mathbb{F}_2^k} \mathbb{E}_{\Lambda(\mathbf{a})}W_{\Lambda(\mathbf{a})}(\mathbf{y})\left ( W_{\Lambda(\mathbf{a})}(\mathbf{y}) + W_{\Lambda(\mathbf{a})+\mathbf{1}}(\mathbf{y})  +   \frac{1}{2^{t-1} - 1} \sum_{\substack{w(\mathbf{c}) \equiv 0 \pmod{2}\\ \mathbf{c}\neq 0, \mathbf{c}\neq 1}}  W_{\Lambda(\mathbf{a})+\mathbf{c}}(\mathbf{y})  \right )^{\alpha-1}\\
&\leq\frac{2^{n(\alpha-1)}}{2^{k\alpha}}\sum_{\mathbf{y}\in \mathbb{F}_2^n}  \sum_{\mathbf{a}\in \mathbb{F}_2^k} \mathbb{E}_{\Lambda(\mathbf{a})}W_{\Lambda(\mathbf{a})}(\mathbf{y})\left ( W_{\Lambda(\mathbf{a})}(\mathbf{y}) + W_{\Lambda(\mathbf{a})+\mathbf{1}}(\mathbf{y})  +  \frac{1}{2^{t-1} - 1} \sum_{w(\mathbf{c}) \equiv 0 \pmod{2}}  W_{\Lambda(\mathbf{a})+\mathbf{c}}(\mathbf{y})   \right )^{\alpha-1}\\
\intertext{ By using $(x+y+z)^{\alpha-1}\leq x^{\alpha-1}+y^{\alpha-1}+z^{\alpha-1}$ }
&\leq\frac{2^{n(\alpha-1)}}{2^{k\alpha}}\sum_{\mathbf{y}\in \mathbb{F}_2^n}  \sum_{\mathbf{a}\in \mathbb{F}_2^k} \mathbb{E}_{\Lambda(\mathbf{a})}W_{\Lambda(\mathbf{a})}(\mathbf{y})^{\alpha} + \frac{2^{n(\alpha-1)}}{2^{k\alpha}}\sum_{\mathbf{y}\in \mathbb{F}_2^n}  \sum_{\mathbf{a}\in \mathbb{F}_2^k} \mathbb{E}_{\Lambda(\mathbf{a})}W_{\Lambda(\mathbf{a})}(\mathbf{y})W_{\Lambda(\mathbf{a})+\mathbf{1}}(\mathbf{y})^{\alpha-1}      \\ &\hspace{2cm}+\frac{2^{n(\alpha-1)}}{2^{k\alpha}}\sum_{\mathbf{y}\in \mathbb{F}_2^n}  \sum_{\mathbf{a}\in \mathbb{F}_2^k} \mathbb{E}_{\Lambda(\mathbf{a})} (\frac{1}{2^{t-1} - 1})^{\alpha-1}W_{\Lambda(\mathbf{a})}(\mathbf{y})\left(  \sum_{w(\mathbf{c}) \equiv 0 \pmod{2}}  W_{\Lambda(\mathbf{a})+\mathbf{c}}(\mathbf{y}) \right)^{\alpha-1}\\
&= 2^{2(\alpha-1) t(1-R-\frac{H_\alpha(W)}{n})} + 2^{2(\alpha-1)t(1-R-\frac{H'_\alpha(W)}{n})} \\ & \hspace{2cm}  + (\frac{1}{2^{t-1} - 1})^{\alpha-1}\frac{2^{n(\alpha-1)}}{2^{k\alpha}}\sum_{\mathbf{y}\in \mathbb{F}_2^n}  \sum_{\mathbf{a}\in \mathbb{F}_2^k} \frac{1}{2^n} \left[ \left(  \sum_{w(\mathbf{c}) \equiv 0 \pmod{2}}  W_{\mathbf{c}}(\mathbf{y}) \right)^{\alpha} + \left(  \sum_{w(\mathbf{c}) \equiv 1 \pmod{2}}  W_{\mathbf{c}}(\mathbf{y}) \right)^{\alpha} \right]\\
\end{align*}

If we let $\sum_{w(\mathbf{c}) \equiv 0 \pmod{2}}  W_{\mathbf{c}}(\mathbf{y}) < 1$ and $\sum_{w(\mathbf{c}) \equiv 1 \pmod{2}}  W_{\mathbf{c}}(\mathbf{y}) < 1$, this could be simplified as 
    $$ 2^{2(\alpha-1) t(1-R-\frac{H_\alpha(W)}{n})} + 2^{2(\alpha-1)t(1-R-\frac{H'_\alpha(W)}{n})}  + 2\cdot\left(\frac{ 2^{2t}}{2^t(2^{t-1}-1)}\right)^{\alpha-1}.$$
    The limitation goes to $1$ when $1-R-\frac{H_\alpha(W)}{n}<0$ if we first let $t\rightarrow\infty$ and then $\alpha\rightarrow 1$. Here $1-R-\frac{H'_\alpha(W)}{n}<0$ is satisfied by rearrangement inequality.
    
If we have additional requirements such that $\sum_{w(\mathbf{c}) \equiv 0 \pmod{2}}  W_{\mathbf{c}}(\mathbf{y}) = \sum_{w(\mathbf{c}) \equiv 1 \pmod{2}}  W_{\mathbf{c}}(\mathbf{y}) = \frac{1}{2}$. Then for $\alpha\in (1,2)$, the limitation goes to 1.   
By removing the r.v. G, and the final result proven.
    
\end{proof}

\section{Average to Average Case Reduction}
\label{sec:Average_to_Average_Case_Reduction}

Random code smoothing bound can be applied in reduction between LPN problems of different parameters \cite{debris2022worst}.

\begin{theorem}[\cite{debris2022worst} Theorem 6.1]
    Let $\varepsilon, \eta \in (0, 1)$, $C > 0$ and $n, k, t \in \mathbb{N}$ be such that (for some constant) $\frac{k}{n} = o(1)$ and 
$$2 \ln2 \frac{1 + \eta}{1 - \varepsilon} \frac{1}{\log_2(\frac{n}{k})} \frac{k}{n}t = C \log_2(n)$$
Suppose that there exists an algorithm $A$ which solves $aDP\left(n_a, k, \frac{(1-\varepsilon)n_a}{2}(1 - \frac{1}{n^{C(1+o(1))}})\right)$ with success probability $\varepsilon$. Then there exists an algorithm which solves $aDP(n, k, t)$ with probability bigger than $\Omega\left( \frac{\varepsilon}{\sqrt{n_a}} \right) - n_a2^{- \Omega(n)}$.
\end{theorem}

The original proof in \cite{debris2022worst} did not show use Bernoulli noise in reduction directly but instead use a reduction of ball noise. Here we show how to use Bernoulli noise to directly obtain the same reduction results.
\begin{align*} 
&\mathbb{E}_{\mathbf{G}_{k\times n}\sim \Lambda} 2^{(\alpha-1) D_\alpha(\mathbf{rG}^{\top}  , \langle \mathbf{r}, \mathbf{t} \rangle||U_{\mathbb{F}_2^k},Ber_p)}\\
&= \mathbb{E}_{\mathbf{G}_{k\times n}\sim \Lambda} 2^{(\alpha-1) D_\alpha(\mathbf{r}+U_\mathbf{c^\perp}  , \langle \mathbf{r}, \mathbf{t} \rangle||U_{\mathbb{F}_2^n},Ber_p)} \\
&= \mathbb{E}_{\mathbf{G}_{ (n-k)\times n}\sim \Lambda} 2^{(\alpha-1) D_\alpha(\mathbf{r}+U_\mathbf{c}  , \langle \mathbf{r}, \mathbf{t} \rangle||U_{\mathbb{F}_2^n},Ber_p)} \\
&= \sum_{y \in \mathbb{F}_2^n}\mathbb{E}_{\mathbf{G}} \frac{P(\langle \mathbf{r}, \mathbf{t} \rangle = 1)^\alpha P(\mathbf{r} + U_C = \mathbf{y} | \langle \mathbf{r}, \mathbf{t} \rangle = 1)^\alpha}{\left(\frac{1}{2^n}\right)^{\alpha-1}p^{\alpha-1}} + \sum_{y \in \mathbb{F}_2^n}\mathbb{E}_{\mathbf{G}}\frac{P(\langle \mathbf{r}, \mathbf{t} \rangle = 0)^\alpha P(\mathbf{r} + U_C = \mathbf{y} | \langle \mathbf{r}, \mathbf{t} \rangle = 0)^\alpha}{\left(\frac{1}{2^n}\right)^{\alpha-1}(1-p)^{\alpha-1}}\\ 
 \intertext{ Let $P(\langle \mathbf{r}, \mathbf{t} \rangle=1)=p$ and $\mathbf{r}\sim Ber_r^{\otimes n} $, i.e.  $p=\frac{1 - (1 - 2r)^{|\mathbf{t}|}}{2}$}
 &= \sum_{y \in \mathbb{F}_2^n} p \cdot \mathbb{E}_{\mathbf{G}_{ (n-k)\times n}\sim \Lambda} \frac{ P(\mathbf{r} + U_C = \mathbf{y} | \langle \mathbf{r}, \mathbf{t} \rangle = 1)^\alpha}{\left(\frac{1}{2^n}\right)^{\alpha-1}} + \sum_{y \in \mathbb{F}_2^n}(1-p)\cdot \mathbb{E}_{\mathbf{G}_{ (n-k)\times n}\sim \Lambda} \frac{ P(\mathbf{r} + U_C = \mathbf{y} | \langle \mathbf{r}, \mathbf{t} \rangle = 0)^\alpha}{\left(\frac{1}{2^n}\right)^{\alpha-1}}\\
  \intertext{ Notice that value of $ \langle \mathbf{r}, \mathbf{t} \rangle $ is a requirement on  $ \mathbf{r}$, denote them as $\mathbf{r}_0, \mathbf{r}_1$ separately}
  &= \sum_{y \in \mathbb{F}_2^n} p \cdot \mathbb{E}_{\mathbf{G}_{ (n-k)\times n}\sim \Lambda}\frac{ P(\mathbf{r}_0 + U_C = \mathbf{y} )^\alpha}{\left(\frac{1}{2^n}\right)^{\alpha-1}} + \sum_{y \in \mathbb{F}_2^n}(1-p)\cdot \mathbb{E}_{\mathbf{G}_{ (n-k)\times n}\sim \Lambda}\frac{ P(\mathbf{r}_1 + U_C = \mathbf{y} )^\alpha}{\left(\frac{1}{2^n}\right)^{\alpha-1}}\\
  &\leq  p\cdot ( 2^{(\alpha-1) n(1-\frac{n-k}{n}-\frac{H_\alpha( \mathbf{r}_0)}{n})}  + 1) + (1-p)\cdot ( 2^{(\alpha-1) n(1-\frac{n-k}{n}-\frac{H_\alpha( \mathbf{r}_1)}{n})}  + 1)\\
  &=  p\cdot  2^{(\alpha-1) n(1-\frac{n-k}{n}-\frac{H_\alpha( \mathbf{r}_0)}{n})}   + (1-p)\cdot 2^{(\alpha-1) n(1-\frac{n-k}{n}-\frac{H_\alpha( \mathbf{r}_1)}{n})}  + 1.\\
\end{align*} The limitation goes to $1$ as $1-\frac{n-k}{n}-\frac{H_\alpha( \mathbf{r}_0)}{n}<0$ and $1-\frac{n-k}{n}-\frac{H_\alpha( \mathbf{r}_1)}{n}<0$.

Here we only compute $\frac{H_\alpha( \mathbf{r}_0)}{n}$ and $\frac{H_\alpha( \mathbf{r}_1)}{n}$ in the case when $\alpha=1$, i.e. the entropy rate of $\mathbf{r}_0$ and $\mathbf{r}_1$. 
\begin{lemma}[Entropy Rate of $\mathbf{r}_0$ and $\mathbf{r}_1$]
    $$\frac{H( \mathbf{r}_0)}{n}=\frac{H( \mathbf{r}_1)}{n}=(1-\frac{1}{n})h(r)$$
    where $h(r)=-r\log r- (1-r)\log (1-r)$.
\end{lemma}
\begin{proof}
    Denote $n$ dimentional random variable $\mathbf{r}_1 = (X_1, X_2, ..., X_n)$, where each $X_i\sim Ber_r$. Since $ \langle \mathbf{r}_1, \mathbf{t} \rangle = 1$, $\mathbf{r}_1$ should have odd 1's in positions where $\mathbf{t}$ has $1$. Thus by reordering the positions, we will get $$\mathbf{r}_1 = (X_1', X_2', ..., X_{n-|\mathbf{t}|}', Y_1, Y_2, ..., Y_{|\mathbf{t}|-1}, Y_{|\mathbf{t}|}) = (X_1', X_2', ..., X_{n-|\mathbf{t}|}', Y_1, Y_2, ..., Y_{|\mathbf{t}|-1}, 1-Y_1- Y_2- ...- Y_{|\mathbf{t}|-1}).$$
    Thus $H( \mathbf{r}_1) = \sum_{i=1}^{n-|\mathbf{t}|}H(X_i')+\sum_{i=1}^{|\mathbf{t}|-1}H(Y_i)=(n-1)h(r)$.    
    Similarly by reordering we will get 
    $$\mathbf{r}_0 = (X_1', X_2', ..., X_{n-|\mathbf{t}|}', Y_1, Y_2, ..., Y_{|\mathbf{t}|-1}, Y_{|\mathbf{t}|}) = (X_1', X_2', ..., X_{n-|\mathbf{t}|}', Y_1, Y_2, ..., Y_{|\mathbf{t}|-1}, Y_1+ Y_2+ ...+ Y_{|\mathbf{t}|-1}).$$
    And $H( \mathbf{r}_0) = \sum_{i=1}^{n-|\mathbf{t}|}H(X_i')+\sum_{i=1}^{|\mathbf{t}|-1}H(Y_i)=(n-1)h(r)$. 
\end{proof}
Thus the optimal bound is $\frac{k}{n}=h(r)  \Rightarrow  r=h^{-1}(\frac{k}{n})=\frac{\frac{k}{n}}{-\log_2\frac{k}{n}}(1+o(1))\Rightarrow .$ 
\begin{align*}
\frac{1}{2}-\frac{1}{2}(1-2r)^{|t|}&= \frac{1}{2}-(\frac{1}{2})^{1+|t|\log_2(1-2r)}\\
    &= \frac{1}{2}-(\frac{1}{2})^{1+|t|\log_2\left(  1+2\frac{\frac{k}{n}}{\log_2\frac{k}{n}}(1+o(1))\right)}\\
    &= \frac{1}{2}-  (\frac{1}{2})^{1+|t|\frac{2}{\ln2}  \frac{\frac{k}{n}}{\log_2\frac{k}{n}}(1+o(1))}  \\
    &=\frac{1}{2} - \frac{1}{2n^{C(1+o(1))}},\\
\end{align*} which completes the reduction proof.

\section{Conclusion}
This paper delves into the critical role of the smoothing parameter and bound, a key concept involves adding sufficient noise to a discrete structure, such as a code, to make its distribution approximate uniformity over the Hamming space. We explore the optimization of this parameter using Rényi divergence, which includes deriving the smoothing bound for random linear codes across all Rényi parameters $\alpha \in (1, \infty)$. We further refine this analysis by reducing random linear codes to classes of random self-dual and quasi-cyclic codes, enhancing their structural properties. An application of our results demonstrates an average-case to average-case reduction from LPN to the average-case decoding problem, utilizing Rényi divergence and focusing on Bernoulli noise. 
%Bibliography
\bibliographystyle{alpha}  
\bibliography{references}

\newcommand{\etalchar}[1]{$^{#1}$}
\begin{thebibliography}{BLRL{\etalchar{+}}18}

\bibitem[ABB{\etalchar{+}}22]{aragon2022bike}
Nicolas Aragon, Paulo Barreto, Slim Bettaieb, Loic Bidoux, Olivier Blazy,
  Jean-Christophe Deneuville, Philippe Gaborit, Santosh Ghosh, Shay Gueron, Tim
  G{\"u}neysu, et~al.
\newblock Bike: bit flipping key encapsulation.
\newblock 2022.

\bibitem[Ale03]{alekhnovich2003more}
Michael Alekhnovich.
\newblock More on average case vs approximation complexity.
\newblock In {\em 44th Annual IEEE Symposium on Foundations of Computer
  Science, 2003. Proceedings.}, pages 298--307. IEEE, 2003.

\bibitem[BCL{\etalchar{+}}17]{bernstein2017classic}
Daniel~J Bernstein, Tung Chou, Tanja Lange, Ingo von Maurich, Rafael Misoczki,
  Ruben Niederhagen, Edoardo Persichetti, Christiane Peters, Peter Schwabe,
  Nicolas Sendrier, et~al.
\newblock Classic mceliece: conservative code-based cryptography.
\newblock {\em NIST submissions}, 1(1):1--25, 2017.

\bibitem[BLP{\etalchar{+}}13]{brakerski2013classical}
Zvika Brakerski, Adeline Langlois, Chris Peikert, Oded Regev, and Damien
  Stehl{\'e}.
\newblock Classical hardness of learning with errors.
\newblock In {\em Proceedings of the forty-fifth annual ACM symposium on Theory
  of computing}, pages 575--584, 2013.

\bibitem[BLRL{\etalchar{+}}18]{bai2018improved}
Shi Bai, Tancr{\`e}de Lepoint, Adeline Roux-Langlois, Amin Sakzad, Damien
  Stehl{\'e}, and Ron Steinfeld.
\newblock Improved security proofs in lattice-based cryptography: using the
  r{\'e}nyi divergence rather than the statistical distance.
\newblock {\em Journal of Cryptology}, 31:610--640, 2018.

\bibitem[BLVW19]{brakerski2019worst}
Zvika Brakerski, Vadim Lyubashevsky, Vinod Vaikuntanathan, and Daniel Wichs.
\newblock Worst-case hardness for lpn and cryptographic hashing via code
  smoothing.
\newblock In {\em Annual international conference on the theory and
  applications of cryptographic techniques}, pages 619--635. Springer, 2019.

\bibitem[BMPS20]{biasse2020less}
Jean-Fran{\c{c}}ois Biasse, Giacomo Micheli, Edoardo Persichetti, and Paolo
  Santini.
\newblock Less is more: code-based signatures without syndromes.
\newblock In {\em Progress in Cryptology-AFRICACRYPT 2020: 12th International
  Conference on Cryptology in Africa, Cairo, Egypt, July 20--22, 2020,
  Proceedings 12}, pages 45--65. Springer, 2020.

\bibitem[CPW]{Chen_Peterson_Weldon}
C.L. Chen, W.W. Peterson, and E.J. Weldon.
\newblock Some results on quasi-cyclic codes.
\newblock {\em Information and Control}, page 407–423.

\bibitem[Cve12]{cvetkovski2012inequalities}
Zdravko Cvetkovski.
\newblock {\em Inequalities: theorems, techniques and selected problems}.
\newblock Springer Science \& Business Media, 2012.

\bibitem[DAR22]{debris2022worst}
Thomas Debris-Alazard and Nicolas Resch.
\newblock Worst and average case hardness of decoding via smoothing bounds.
\newblock {\em Cryptology ePrint Archive}, 2022.

\bibitem[DJ24]{dao2024lossyDensesparseLPN}
Quang Dao and Aayush Jain.
\newblock Lossy cryptography from code-based assumptions.
\newblock {\em arXiv preprint arXiv:2402.03633}, 2024.

\bibitem[DP82]{delsarte1982algebraic}
Philippe Delsarte and Philippe Piret.
\newblock Algebraic constructions of shannon codes for regular channels.
\newblock {\em IEEE Transactions on Information Theory}, 28(4):593--599, 1982.

\bibitem[DvW22]{ducas2022lattice}
L{\'e}o Ducas and Wessel van Woerden.
\newblock On the lattice isomorphism problem, quadratic forms, remarkable
  lattices, and cryptography.
\newblock In {\em Annual International Conference on the Theory and
  Applications of Cryptographic Techniques}, pages 643--673. Springer, 2022.

\bibitem[Hay06]{hayashi2006general}
Masahito Hayashi.
\newblock General nonasymptotic and asymptotic formulas in channel
  resolvability and identification capacity and their application to the
  wiretap channel.
\newblock {\em IEEE Transactions on Information Theory}, 52(4):1562--1575,
  2006.

\bibitem[Hay11]{hayashi2011exponential}
Masahito Hayashi.
\newblock Exponential decreasing rate of leaked information in universal random
  privacy amplification.
\newblock {\em IEEE Transactions on Information Theory}, 57(6):3989--4001,
  2011.

\bibitem[HKL{\etalchar{+}}12]{heyse2012lapin}
Stefan Heyse, Eike Kiltz, Vadim Lyubashevsky, Christof Paar, and Krzysztof
  Pietrzak.
\newblock Lapin: an efficient authentication protocol based on ring-lpn.
\newblock In {\em Fast Software Encryption: 19th International Workshop, FSE
  2012, Washington, DC, USA, March 19-21, 2012. Revised Selected Papers}, pages
  346--365. Springer, 2012.

\bibitem[HM16]{hayashi2016secure}
Masahito Hayashi and Ryutaroh Matsumoto.
\newblock Secure multiplex coding with dependent and non-uniform multiple
  messages.
\newblock {\em IEEE Transactions on Information Theory}, 62(5):2355--2409,
  2016.

\bibitem[HV93]{han1993approximation}
Te~Sun Han and Sergio Verd{\'u}.
\newblock Approximation theory of output statistics.
\newblock {\em IEEE Transactions on Information Theory}, 39(3):752--772, 1993.

\bibitem[Loe94]{loeliger1994basic}
Hans-Andrea Loeliger.
\newblock On the basic averaging arguments for linear codes.
\newblock {\em Communications and Cryptography: Two Sides of One Tapestry},
  pages 251--261, 1994.

\bibitem[MAB{\etalchar{+}}18]{melchor2018HQC}
Carlos~Aguilar Melchor, Nicolas Aragon, Slim Bettaieb, Lo{\i}c Bidoux, Olivier
  Blazy, Jean-Christophe Deneuville, Philippe Gaborit, Edoardo Persichetti,
  Gilles Z{\'e}mor, and IC~Bourges.
\newblock Hamming quasi-cyclic (hqc).
\newblock {\em NIST PQC Round}, 2(4):13, 2018.

\bibitem[McE78]{mceliece1978public}
Robert~J McEliece.
\newblock A public-key cryptosystem based on algebraic.
\newblock {\em Coding Thv}, 4244:114--116, 1978.

\bibitem[MR07]{micciancio2007worst}
Daniele Micciancio and Oded Regev.
\newblock Worst-case to average-case reductions based on gaussian measures.
\newblock {\em SIAM Journal on Computing}, 37(1):267--302, 2007.

\bibitem[MST72]{macwilliams1972good}
F.~Jessie MacWilliams, Neil J.~A. Sloane, and John~G. Thompson.
\newblock Good self dual codes exist.
\newblock {\em Discrete Mathematics}, 3(1-3):153--162, 1972.

\bibitem[PB23]{pathegama2023smoothing}
Madhura Pathegama and Alexander Barg.
\newblock Smoothing of binary codes, uniform distributions, and applications.
\newblock {\em Entropy}, 25(11):1515, 2023.

\bibitem[PB24]{pathegama2024r}
Madhura Pathegama and Alexander Barg.
\newblock R$\backslash$'enyi divergence guarantees for hashing with linear
  codes.
\newblock {\em arXiv preprint arXiv:2405.04406}, 2024.

\bibitem[Pei09]{peikert2009public}
Chris Peikert.
\newblock Public-key cryptosystems from the worst-case shortest vector problem.
\newblock In {\em Proceedings of the forty-first annual ACM symposium on Theory
  of computing}, pages 333--342, 2009.

\bibitem[Pre17]{prest2017sharper}
Thomas Prest.
\newblock Sharper bounds in lattice-based cryptography using the r{\'e}nyi
  divergence.
\newblock In {\em Advances in Cryptology--ASIACRYPT 2017: 23rd International
  Conference on the Theory and Applications of Cryptology and Information
  Security, Hong Kong, China, December 3-7, 2017, Proceedings, Part I 23},
  pages 347--374. Springer, 2017.

\bibitem[Reg09]{regev2009lattices}
Oded Regev.
\newblock On lattices, learning with errors, random linear codes, and
  cryptography.
\newblock {\em Journal of the ACM (JACM)}, 56(6):1--40, 2009.

\bibitem[YT18]{yu2018renyi}
Lei Yu and Vincent~YF Tan.
\newblock R{\'e}nyi resolvability and its applications to the wiretap channel.
\newblock {\em IEEE Transactions on Information Theory}, 65(3):1862--1897,
  2018.

\bibitem[Yu24]{yu2024renyiresolvability}
Lei Yu.
\newblock R$\backslash$'enyi resolvability, noise stability, and
  anti-contractivity.
\newblock {\em arXiv preprint arXiv:2402.07660}, 2024.

\bibitem[YZ21]{yu2021smoothing}
Yu~Yu and Jiang Zhang.
\newblock Smoothing out binary linear codes and worst-case sub-exponential
  hardness for lpn.
\newblock In {\em Advances in Cryptology--CRYPTO 2021: 41st Annual
  International Cryptology Conference, CRYPTO 2021, Virtual Event, August
  16--20, 2021, Proceedings, Part III 41}, pages 473--501. Springer, 2021.

\end{thebibliography}

\end{document}